\documentclass[a4paper,UKenglish]{lipics-v2018}

\usepackage{microtype}%if unwanted, comment out or use option "draft"
\nolinenumbers
\usepackage{hyperref}
\usepackage{amsmath,amsthm, amssymb}
\usepackage{stmaryrd}

\usepackage{tikz}
\usetikzlibrary{intersections}
\usetikzlibrary{decorations.markings}
\usetikzlibrary{arrows, automata, positioning,calc}
\usetikzlibrary{decorations.pathmorphing}
\usetikzlibrary{decorations.pathreplacing}
\usepackage{3dplot}
\usepackage{comment}
\usepackage{color}
\usepackage{gastex}
\usepackage{algorithmic}
\usepackage{algorithm}
\usepackage{mathabx} % for the blacktriangleright, added 25.03.2016

\newcommand{\RP}{\mathbb{R}_{\ge 0}}
\newcommand{\R}{\mathbb{R}}
\newcommand{\N}{\mathbb{N}}
\newcommand{\Z}{\mathbb{Z}}

\newcommand{\locs}{L} % location set 
\newcommand{\loc}{\ell} % one location 
\newcommand{\clocks}{\mathcal{X}} % clock variables
\newcommand{\guard}{\varphi} 
\newcommand{\reset}{\lambda}
\newcommand{\val}{\nu} % clock valuation
\newcommand{\cost}{{\mathrm{cost}}}

\newcommand{\vv}{\boldsymbol{v}}
\newcommand{\uu}{\boldsymbol{u}}
\newcommand{\true}{\texttt{true}} % TRUE

\newcommand{\tuple}[1]{\langle #1 \rangle}

\newcommand{\zeroval}{\boldsymbol{0}} % parameters

\newcommand{\edges}{E}

\makeatletter
\newcommand*{\defeq}{\mathrel{\rlap{%
                     \raisebox{0.3ex}{$\m@th\cdot$}}%
                     \raisebox{-0.3ex}{$\m@th\cdot$}}%
                     =}
\makeatother

\makeatletter
\newcommand*{\ndefeq}{\mathrel{\rlap{%
                     \raisebox{0.3ex}{$\m@th\cdot$}}%
                     \raisebox{-0.3ex}{$\m@th\cdot$}%
                     \rlap{%
                     \raisebox{0.3ex}{$\m@th\cdot$}}
                     \raisebox{-0.3ex}{$\m@th\cdot$}}
                     =}
\makeatother

\newtheorem{proposition}[theorem]{Proposition}
\newcommand{\mygrid}[4]
{
\draw[very thin,color=gray,step=.25cm] (#1,#2) grid (#3-.1,#4-.1);
\draw[->,ultra thick] (#1-.1,0)--(#3,0) node[right]{$x$};
\draw[->,ultra thick] (0,#2-.1)--(0,#4) node[above]{$y$};
}

\title{Costs and  Rewards in Priced Timed Automata}

\titlerunning{Costs and Rewards in MPTA} %optional, in case that the title is too long; the running title should fit into the top page column

\author{Martin Fr\"anzle}{Department of Computing Science, University of Oldenburg, Germany}{martin.fraenzle@informatik.uni-oldenburg.de}{}{}
\author{Mahsa Shirmohammadi}{CNRS \& LIS, France}{mahsa.shirmohammadi@lis-lab.fr}{}{}
\author{Mani Swaminathan}{Department of Computing Science, University of Oldenburg, Germany}{mani.swaminathan@informatik.uni-oldenburg.de}{}{}
\author{James Worrell}{Department of Computer Science, University of Oxford, UK}{james.worrell@cs.ox.ac.uk}{}{}

\authorrunning{M. Fr\"anzle, M. Shirmohammadi, M. Swaminathan, J. Worrell} %mandatory. First: Use abbreviated first/middle names. Second (only in severe cases): Use first author plus 'et. al.'

\Copyright{Martin Fr\"anzle, Mahsa Shirmohammadi, Mani Swaminathan, and James Worrell}%mandatory, please use full first names. LIPIcs license is "CC-BY";  http://creativecommons.org/licenses/by/3.0/

\subjclass{Theory of computation: Timed and hybrid models }% mandatory: Please choose ACM 1998 classifications from http://www.acm.org/about/class/ccs98-html . E.g., cite as "F.1.1 Models of Computation". 
% mandatory: Please provide 1-5 keywords

\relatedversion{A shorter version of this paper appears in the proceedings of ICALP 2018.}

\acknowledgements{to Dan Segal for discussions on~\cite{GrunewaldSegal04} and the reviewers for their feedback.}

\EventEditors{Ioannis Chatzigiannakis, Christos Kaklamanis, D\'{a}niel Marx, and Don Sannella}
\EventNoEds{4}
\EventLongTitle{45th International Colloquium on Automata, Languages, and Programming (ICALP 2018)}
\EventShortTitle{ICALP 2018}
\EventAcronym{ICALP}
\EventYear{2018}
\EventDate{July 9--13, 2018}
\EventLocation{Prague, Czech Republic}
\EventLogo{eatcs}
\SeriesVolume{107}
\ArticleNo{248}
\begin{document}

\maketitle

%==========abstract============
\begin{abstract}
We consider Pareto analysis of reachable states of multi-priced timed automata (MPTA):
 timed automata equipped with multiple  observers that keep track of costs (to be minimised) and rewards (to
  be maximised) along a computation.  Each observer has a constant
  non-negative derivative which may depend on the location of the MPTA.

  We study the Pareto Domination Problem, which asks
  whether it is possible to reach a target location via a run in
  which the accumulated costs and rewards Pareto dominate a given
  objective vector.  We show that this problem is undecidable in general,
  but decidable for MPTA with at most three observers.  For
  MPTA whose observers are all costs or all rewards, we show that the Pareto Domination Problem is PSPACE-complete.  
  We also consider an
  $\varepsilon$-approximate Pareto Domination Problem that is decidable without restricting
  the number and types of observers.

  We develop connections between MPTA and Diophantine
  equations. Undecidability of the Pareto Domination Problem is shown by reduction from
  Hilbert's $10^\text{th}$ Problem, while decidability for three
  observers is shown by a translation to a fragment of arithmetic
  involving quadratic forms.
\end{abstract}

%==========introduction============
\section{Introduction}\label{sec:intro}
\emph{Multi Priced Timed Automata}
(MPTA)~\cite{Bouyer08b,Bouyer08c,BrihayeBruyereRaskin06,FORMATS09,Larsen08,Perevoshchikov15,Quaas10}
extend priced timed
automata~\cite{Alur01PTA,Behrmann01,Bouyer07,Bouyer08d,Larsen01} with
\emph{multiple observers} that capture the accumulation of costs and
rewards along a computation. This extension allows to model multi-objective
optimization problems beyond the scope of timed automata~\cite{AlurDill94}. 
MPTA lie at the frontier between timed automata (for which
reachability is decidable~\cite{AlurDill94}) and linear hybrid automata (for which
reachability is undecidable~\cite{Henzinger98}). The observers
exhibit richer dynamics than the clocks of timed automata by not being confined to
unit slope in locations, but may neither be queried nor reset while
taking edges.  This \emph{observability restriction} has been
exploited in~\cite{Larsen08} (under a cost-divergence assumption) for carrying out a
\emph{Pareto analysis} of reachable values of the observers.

In this paper we distinguish
between observers that represent \emph{costs} (to be minimised) and
those that represent \emph{rewards} (to be maximised).  Formally, we
partition the set $\mathcal{Y}$ of observers into cost and reward
variables and say that $\gamma\in \RP^{\mathcal{Y}}$ \emph{Pareto
  dominates} $\gamma'\in \RP^{\mathcal{Y}}$ if
$\gamma(y)\leq\gamma'(y)$ for each cost variable $y$ and
$\gamma(y)\geq\gamma'(y)$ for each reward variable $y$.  Then the
\emph{Pareto curve} corresponding to an MPTA consists of all
undominated vectors $\gamma$ that are reachable in an accepting location.  
While cost and reward variables are syntactically identical in the underlying automaton model, distinguishing between them
%(cf. Section~\ref{sec:definitions}) 
changes the notion of Pareto
domination and the associated decision problems.

We introduce in Section~\ref{sec:definitions} a decision version of the problem of computing Pareto curves
for MPTA, called the \emph{Pareto Domination Problem}.  Here, given a
target vector $\gamma\in \RP^{\mathcal{Y}}$, one asks to reach an
accepting location with a valuation $\gamma'\in \RP^{\mathcal{Y}}$ that Pareto
dominates~$\gamma$.  This has not been addressed in prior
work on Pareto analysis of MPTA~\cite{Larsen08}, which considers only costs or only
rewards. Other works on MPTA either do not address Pareto
analysis~\cite{Bouyer08b,BrihayeBruyereRaskin06,FORMATS09,Perevoshchikov15,Quaas10},
 or have only discrete costs updated on edges~\cite{Zhang17},
or are confined to a single clock~\cite{Bouyer08c}.

Our first main result is that the Pareto
Domination Problem is undecidable in general.  The undecidability
proof in Section~\ref{sec:undecidability} is by reduction from
Hilbert's $10^\text{th}$ problem.  Owing to the existence of so-called
``universal Diophantine equations'' (of degree 4 with 58 variables~\cite{Jones80}), our proof shows undecidabililty of the Pareto
Domination Problem for some fixed but large number of observers.
Undecidability of the Pareto Domination Problem entails that one
cannot compute an exact Pareto curve for an arbitrary MPTA. %\textcolor{red}{This
%strengthens the undecidability result in~\cite{FORMATS09} for location
%reachability in MPTA with upper and lower bounds on all non-negative
%and negative observers along all feasible runs of the MPTA.}
%\marginpar{Mahsa: the red sentence is confusing to me}
%Swami: I have after a second reading removed the sentence, as FORMATS09 does not consider PDP.

We consider three different approaches to recover decidability of the
Pareto Domination Problem, which all have a common foundation, namely
a \emph{monotone} VASS described in Sections~\ref{sec:prelims} and \ref{sec:simplex-automaton}, which simulates integer runs of a given MPTA. By analysing the semi-linear reachability set of this VASS we can
reduce the Pareto Domination Problem to satisfiability of a class of
bilinear mixed integer-real constraints.  We then consider
restrictions on MPTA and variants of the Pareto Domination Problem
that allow us to solve this class of constraints.

We first show in Section~\ref{sec:pure-const} that restricting to
MPTA with only costs or only rewards yields PSPACE-completeness of the
Pareto Domination Problem.  Here we are able to eliminate integer
variables from our bilinear constraints, resulting in a formula
of linear real arithmetic.  This strengthens~\cite[Theorem 1 and
Corollary 1]{Larsen08}, whose decision
procedures %under a cost-divergence assumption on cycles and
(that exploit well-quasi-orders  for termination) do not yield complexity bounds.
%\marginpar{Mahsa: I am in favor of WQO termination, as it shows we improve the results hugely}
%Swami: agreed. WQO-reference retained and \marginpar removed.

Next we confine the MPTA in Section~\ref{sec:three-cost-variables} to
at most three observers, but allow a mix of costs and rewards.
Decidability is now achieved by eliminating real variables from the
bilinear constraint system, thus reducing the Pareto Domination
Problem to deciding the existence of positive integer zeros of a quadratic
form, which is known to be decidable from~\cite{GrunewaldSegal04}.

We consider in Section~\ref{sec:gap} another method to restore
decidability for general MPTA with arbitrarily many costs and rewards,
by studying an approximate version of the Pareto Domination Problem,
called the \emph{Gap Domination Problem}.  
Similar to
the setting of~\cite{DiakonikolasY09}, the Gap Domination Problem represents the decision version of the problem of
computing $\varepsilon$-Pareto curves.  This problem, whose 
input
includes a tolerance $\varepsilon>0$ and a vector
$\gamma\in\RP^{\mathcal{Y}}$, 
permits inconclusive answers if all
solutions dominating $\gamma$ do so with a slack of less than~$\varepsilon$.  
We solve the Gap Domination Problem by relaxation and
rounding applied to our bilinear system of constraints.

In this paper we consider only MPTA with non-negative rates.  Our
approach can be generalised to obtain decidability results also in the
case of negative rates by extending our foundation in Sections~\ref{sec:prelims} 
and \ref{sec:simplex-automaton} from monotone VASS to $\mathbb{Z}$-VASS~\cite{HaaseHalfon14}.

\section{Background}\label{sec:prelims}
{\bf Quadratic Diophantine Equations.} For later use we recall a
decidable class of non-linear Diophantine problems.  Consider the
quadratic equation
\begin{gather}
\sum_{i,j=1}^n a_{ij}X_iX_j + \sum_{j=1}^n b_jX_j +c = 0 
\label{eq:quadratic}
\end{gather}
whose coefficients $a_{ij}$, $b_j$, and $c$ are rational numbers.
Consider also the family of constraints
\begin{equation}
  f_1(X_1,\ldots,X_n) \, \sim \, c_1 \wedge  \, \ldots \, \wedge \, f_k(X_1,\ldots,X_n) \, \sim \, c_k \, ,
\label{eq:linear}
\end{equation}
where $f_1,\ldots,f_k$ are linear forms with rational coefficients, $c_1,\ldots,c_k\in \mathbb{Q}$,
and ${\sim}\in \{<,\leq\}$.

\begin{theorem}[\cite{GrunewaldSegal04}]
There is an algorithm that decides whether a given quadratic equation~(\ref{eq:quadratic}) and a family
of linear inequalities~(\ref{eq:linear}) have a solution in $\mathbb{Z}^n$.
\label{thm:segal}
\end{theorem}

Let us emphasize that in Theorem~\ref{thm:segal} at most one quadratic
constraint is permitted.  It is clear (e.g., by
 introducing a slack variable) that the theorem remains true if the
 equality symbol in (\ref{eq:quadratic}) is replaced by any  comparison operator in~$\{<,\leq, >,\geq\}$.

\bigskip

\noindent {\bf Monotone VASS.} 
A \emph{monotone vector addition system with states} (monotone VASS)
is a tuple~$\mathcal{Z} = \tuple{n,Q,q_0,Q_f, \Sigma,\Delta}$, where
$n\in\mathbb{N}$ is the \emph{dimension}, $Q$ is a set of
\emph{states}, $q_0\in Q$ is the \emph{initial state},
$Q_f\subseteq Q$ is a set of \emph{final states}, $\Sigma$ is the set of
\emph{labels}, and
$\Delta \subseteq Q \times \mathbb{N}^n \times \Sigma \times Q$ is the
set of \emph{transitions}.

Given such a monotone VASS $\mathcal{Z}$ as above, the family of sets
$\mathrm{Reach}_{\mathcal{Z},q} \subseteq \mathbb{N}^n$, for $q\in Q$, is the
minimal family (w.r.t. to set inclusion) of integer vectors such that
$\boldsymbol{0} \in \mathrm{Reach}_{\mathcal{Z},q_0}$ and
for all $q\in Q$, if~$\uu \in \mathrm{Reach}_{\mathcal{Z},q}$ and $(q,\vv,\ell,p) \in \Delta$
  for some $\ell\in L$, then $\uu+\vv \in \mathrm{Reach}_{\mathcal{Z},p}$.
Finally we define the \emph{reachability set} of $\mathcal{Z}$ to be
$\mathrm{Reach}_{\mathcal{Z}} := \bigcup_{q\in Q_f} \mathrm{Reach}_{\mathcal{Z},q}$.

For every vector $\vv \in \mathbb{N}^n$ and every finite set
$P=\{\uu_1,\ldots,\uu_m \}$ of vectors in $\mathbb{N}^n$, we define
the \emph{$\mathbb{N}$-linear set}
$S(\vv,P) := \{ \vv + \sum_{i=1}^m a_i \uu_i : a_1,\ldots,a_m \in
\mathbb{N}\}$.  We call $\boldsymbol{v}$ the \emph{base vector} 
and~$\boldsymbol{u}_1,\ldots,\boldsymbol{u}_m \in P$ the \emph{period vectors}
of the set.

The following proposition follows from~\cite[Proposition 4.3]{Lin10},\cite{KopTo10}
(see Appendix~\ref{append-zvass}).
\begin{proposition} \label{prop:zvass} Let
  $\mathcal{Z}=\tuple{n,Q,q_0,Q_f, \Sigma,\Delta}$ be a monotone VASS.
  Then the set $\mathrm{Reach}_{\mathcal{Z}}$ can be written as a
  finite union of $\mathbb{N}$-linear sets
  $S(\vv_1,P_1),\ldots,S(\vv_k,P_k)$, where for $i=1,\ldots,k$ the
  components of $\vv_i$ and of each vector in $P_i$ are bounded by
  $\mathit{poly}(n,|Q|,M)^n$ in absolute value, where $M$ is maximum
  absolute value of the entries of vectors in $\mathbb{N}^n$
  occurring in $\Delta$.
\label{prop:LIN}
\end{proposition}

%==========definitions============
\section{Multi-Priced Timed Automata and Pareto Domination}\label{sec:definitions}

Let $\RP$ denote the set of non-negative real numbers.  Given a
set~$\clocks=\{x_1,\ldots,x_n\}$ of \emph{clocks}, the set
$\Phi(\clocks)$ of \emph{clock constraints}\label{clock-constraint} is
generated by the grammar $ \varphi ::= \true \mid x\leq k \,\mid\,
x\geq k \,\mid\, \varphi \wedge \varphi \, , $ where $k \in \N$ is a
natural number and $x\in \clocks$.  A \emph{clock valuation} is a
mapping~$\val: \clocks \to \RP$ that assigns to each clock a
non-negative real number.  We denote by $\zeroval$ the valuation such
that~$\zeroval(x)=0$ for all clocks $x\in \clocks$.  We write
$\val\models\varphi$ to denote that~$\val$ satisfies the
constraint~$\guard$.  Given $t\in\RP$, we let $\val+t$ be the clock
valuation such that~$(\val+t)(x)=\val(x)+t$ for all
clocks~$x\in\clocks$.  Given $\reset\subseteq\clocks$, let
$\val[\reset\leftarrow 0]$ be the clock valuation such that
$\val[\reset\leftarrow 0](x)=0$ if~$x\in\reset$, and
$\val[\reset\leftarrow 0](x)=\val(x)$ otherwise.

A \emph{multi-priced timed automaton} (MPTA) is a tuple
$\mathcal{A}=\tuple{\locs,\ell_0,L_f,\clocks,\mathcal{Y},\edges,R}$,
where~$\locs$ is a finite set of \emph{locations}, $\ell_0\in L$ is an
\emph{initial location}, $L_f\subseteq L$ is a set of \emph{accepting
  locations}, $\clocks$ is a finite set of \emph{clock variables},
$\mathcal{Y}$ is a finite set of \emph{observers},
$\edges\subseteq \locs\times \Phi(\clocks)\times 2^\clocks\times
\locs$ is the set of \emph{edges},
$R : L \rightarrow \mathbb{N}^{\mathcal{Y}}$ is a \emph{rate
  function}.  Intuitively $R(\loc)$ is a vector that gives the rates
of each observer  in location $\loc$.

A \emph{state} of $\mathcal{A}$ is a triple $(\loc,\nu,t)$ where $\loc$
is a location, $\nu$ a clock valuation, and~$t\in\RP$ is a \emph{time
  stamp}.  A \emph{run} of $\mathcal{A}$ is an alternating sequence of
states and edges
$ \rho = (\ell_0,\nu_0,t_0) \stackrel{e_1}{\longrightarrow}
  (\ell_1,\nu_1,t_1)\stackrel{e_2}{\longrightarrow} \ldots
  \stackrel{e_m}{\longrightarrow} (\ell_m,\nu_m,t_m) \, ,$ where $t_0=0$,
$\nu_0=\boldsymbol{0}$,~$t_{i-1} \le t_i$ for all
$i\in\{1,\ldots,m\}$, and
$e_i =\tuple{\ell_{i-1},\varphi,\lambda,\ell_i} \in E$ is such that
$\nu_{i-1} +(t_i-t_{i-1})\models \varphi$ and~$\nu_i = (\nu_{i-1}+(t_i-t_{i-1}))[\lambda \leftarrow 0]$
for~$i=1,\ldots,m$.  The run is \emph{accepting} if $\ell_m\in L_f$ and said to have 
\emph{granularity} $\frac{1}{g}$ for a fixed $g \in \mathbb{N}$ if all~$t_i \in \mathbb{Q}$ 
are positive integer multiples of~$\frac{1}{g}$.
The \emph{cost} of such a run is a vector $\cost(\rho) \in \mathbb{R}^{\mathcal{Y}}$, defined by
$\cost(\rho) = \sum_{j=0}^{m-1} (t_{i+1}-t_i) R(\ell_i) \, .$

Henceforth we will assume that the set  
$\mathcal{Y}$ of observers  of a given MPTA is partitioned into 
a set $\mathcal{Y}_c$ of \emph{cost variables} and a set $\mathcal{Y}_r$ of 
\emph{reward variables}.  With respect to this
partition we define a \emph{domination ordering} $\preccurlyeq$ on the
set of valuations~$\R^{\mathcal{Y}}$, where
$\gamma \preccurlyeq \gamma'$ if $\gamma(y)\leq\gamma'(y)$ for all
$y\in\mathcal{Y}_r$ and $\gamma'(y)\leq \gamma(y)$ for all
$y\in\mathcal{Y}_c$.  Intuitively $\gamma\preccurlyeq\gamma'$ (read
$\gamma'$ dominates $\gamma$) if $\gamma'$ is at least as good as
$\gamma$ in all respects.

Given $\varepsilon>0$ we define an
\emph{$\varepsilon$-domination ordering} $\preccurlyeq_\varepsilon$, where
$\gamma \preccurlyeq_\varepsilon \gamma'$ (read $\gamma'$
$\varepsilon$-dominates~$\gamma$) if $\gamma(y)+\varepsilon \leq
\gamma'(y)$ for all $y \in \mathcal{Y}_r$ and $\gamma'(y) +
\varepsilon \leq \gamma(y)$ for all $y\in \mathcal{Y}_c$.  We can
think of $\gamma \preccurlyeq_\varepsilon \gamma'$ as denoting that
$\gamma'$ is better than $\gamma$ by an additive factor of $\varepsilon$ in all dimensions.
In particular we clearly have that $\gamma \preccurlyeq_{\varepsilon}
\gamma'$ implies $\gamma \preccurlyeq \gamma'$.

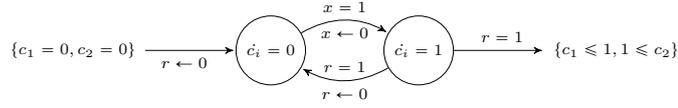
\begin{figure}[t]
\vspace{-.2cm}
    \begin{center}
      ~~~~~~\scalebox{.85}{\begin{tikzpicture}[->,>=stealth',shorten >=1pt,node distance=1cm, initial text={}]
  
\node [state,draw=none] (0,0) (st)  {\scriptsize{$\{c_1=0,c_2=0\}$}};
	\node [state] (q1) [right=1.4cm of st] {\scriptsize{$\dot{c_i}=0$}};
  \node[state] (q2) [right=1.2cm of q1]  {\scriptsize{$\dot{c_i}=1$}};
  \node[state, draw=none] (en) [right=1.4cm of q2]  {\scriptsize{$\{c_1\leq 1, 1\leq c_2\}$}};

  \path[->] (st) edge  % node [above,midway] {\scriptsize{$r= 1~~$}}
  node [below,midway] {\scriptsize{$r\leftarrow 0~~$}} (q1);%
  \path[->] (q1) edge [bend left]  node [above,midway] {\scriptsize{$x=1$}}  node [below,midway] {\scriptsize{$x\leftarrow 0$}} (q2);%
  \path[->] (q2) edge [bend left] node [above,midway] {\scriptsize{$r=1$}}  node [below,midway] {\scriptsize{$r\leftarrow 0$}} (q1);%
  \path[->] (q2) edge  node [above,midway] {\scriptsize{$~r=1$}}  (en);%

\end{tikzpicture}}
    \end{center}
    \vspace{-.6cm}
    \caption{Predicates in curly brackets denote observer values enforced by
      initialisation, $c_i=0$ with $i\in \{1,2\}$, and the Pareto constraint upon
      exit $\{c_1\leq 1, 1\leq c_2\}$. Denoting the initial value of clock $x$ by $x^*$, the value of both
			$c_1$ and $c_2$ after $n$ full traversals of the                                                     
central cycle is $nx^*$.  Meeting the final Pareto constraint from initial
      values  thus requires that $x^*$ be $\frac{1}{n}$ for some
      positive integer~$n$.}\label{fig:exam-init} \vspace{-.4cm}
  \end{figure}

The \emph{Pareto Domination Problem}  is as follows. Given an MPTA
$\mathcal{A}$  with a set $\mathcal{Y}$ of observers and a
partition of $\mathcal{Y}$ into sets $\mathcal{Y}_c$ and
$\mathcal{Y}_r$ of cost and reward variables, with a target~$\gamma\in \R^\mathcal{Y}$, 
decide whether there is an accepting run
$\rho$ of $\mathcal{A}$ such that $\gamma \preccurlyeq \cost(\rho)$.

The \emph{Gap Domination Problem} is a variant of the above problem in
which the input additionally includes an accuracy parameter
$\varepsilon>0$.  If there is some run $\rho$ such that~$\gamma\preccurlyeq_\varepsilon \cost(\rho)$ then the output should be
``dominated'' and if there is no run $\rho$ such that~$\gamma\preccurlyeq \cost(\rho)$ then the output should be ``not
dominated''.  In case neither of these alternatives hold
(i.e., $\gamma$ is dominated but not $\varepsilon$-dominated) then there
is no requirement on the output.

%Given an MPTA $A=\tuple{\locs,\clocks,\mathcal{Y},\edges,R}$, a
%\emph{reachability objective} is a conjunction of atomic constraints
 %of the
%form $y \leq c$ or $y \geq c$ for $y\in\mathcal{Y}$ and $c \in
%\mathbb{Z}_{\geq 0}$.  Such an objective determines a target
%set of values $T\subseteq \RP^{\mathcal{Y}}$ in an obvious way.  We
%assume that the objective has at most one constraint per
%variable.  
%Intuitively each variable $y\in\mathcal{Y}$ represents
%either a cost (to be minimised) or a reward (to be maximised).

In the (Pareto) Domination Problem the objective is to \emph{reach}
an accepting location while satisfying a family of upper-bound constraints on cost variables and
lower-bound constraints on reward variables.  We say that an instance
of the problem is \emph{pure} if all observers are cost variables or
all are reward variables (and hence all constraints are upper bounds
or all are lower bounds); otherwise we call the instance \emph{mixed}.
Our problem formulation involves only
simple constraints on observers, i.e., those of the form $y \leq c$ or $y \geq c$ for
$y\in\mathcal{Y}$.  However such constraints can be used to encode
more general linear constraints of the form~$a_1 y_1+\cdots+ a_ky_k \sim c$, where
$y_1,\ldots,y_k \in \mathcal{Y}$, $a_1,\ldots, a_k, c \in \mathbb{N}$
and ${\sim} \in\{\leq,\geq,=\}$.  To do this one introduces a
fresh observer  to denote each linear term~$a_1 y_1+\cdots+ a_ky_k$
(two fresh observers  are needed for an equality constraint).

%It should be noted that
%unless the number of observer variables is fixed, 
%this does not pose a
%restriction on expressiveness, as mixed upper and lower bound
%constraints on a single observer variable can  be encoded
%by merely duplicating that observer. Similarly, the restriction to
%simple bounds of the form $y \leq c$ or $y \geq c$ does not confine
%expressiveness compared to a setting where bounds on conical 
%combinations of observer variables are permitted, as the particular
%combination can be rendered part of the dynamics of yet another
%observer.

%The reachability problem for MPTAs is as follows.  Given an
%MPTA~$A=\tuple{\locs,\clocks,\mathcal{Y},\edges,R}$ with initial
%location~$\loc_0\in L$, a final location~$\loc_f \in \locs$ and
%reachability objective~$T\subseteq \R^{\mathcal{Y}}$, decide whether
%there is a run~$\rho$ in~$A$ from~$\loc_0$ to~$\loc_f$ such that
%$\cost(\rho)\in T$.

Note that we consider timed automata without \emph{difference
  constraints} on clocks, i.e., without clock guards of the form
  $x_i-x_j \sim k$, for $k \in \mathbb{N}$.  As discussed in
  Appendix~\ref{app:difference} all our decidability and
  complexity results hold also in case of such constraints.

%As summarized in~\cite[Section
%  5.3]{Bouyer07} for the setting of a single observer, given an
%  MPTA~$\mathcal{A}$ with difference clock constraints, we can find an
%  MPTA~$\mathcal{A'}$ without difference clock constraints such that
%  $\mathcal{A}$ and $\mathcal{A'}$ are strongly time-bisimilar.  The
%  Domination Problems for for $\mathcal{A}$ can thus be reduced to
%  those for $\mathcal{A'}$. \marginpar{Swami: Effect on complexity
%  bounds} Eliminating difference clock constraints from MPTA however
%  results in an exponential blow-up in the number of locations and
%  edges, cf. Section 5.3 of~\cite{Bouyer07}.

%==========undecidability============
\section{Undecidability of the Pareto Domination Problem}\label{sec:undecidability}
 %\begin{figure}[t]
  %  \begin{center}
   %   ~~~~~~\scalebox{.85}{\input{Figures/exampleInitial.tex}}
   % \end{center}
		%	\vspace{-.6cm}
		%\caption{Consider this MPTA with initial cost $c=0$. 
		%Writing $x^*$ for the initial value of clock $x$,
		% the total cost accumulated in $c$ is $nx^*$ where $n$ is the number of times the central cycle is traversed. 
		%Thus achieving the final cost $c=1$
		%requires that~$x^*$ be~$\frac{1}{n}$ 
    %   for some positive integer~$n$. 
		%}\label{fig:exam-init}
		%\vspace{-.4cm}
  %\end{figure}

In this section we prove undecidability of the Pareto Domination Problem.
To give some insight we first give in  
Figure~\ref{fig:exam-init} an MPTA, in which
the Pareto constraint $c_1\leq 1, c_2\geq 1$ is used to enforce that when control enters the MPTA the 
value of clock~$x$ is~$\frac{1}{n}$ for some positive integer~$n$.

We prove undecidability of the Pareto Domination Problem by reduction from the
satisfiability problem for a fragment of arithmetic given by a
language $\mathcal{L}$ that is defined as follows.  There is an
infinite family of variables $X_1,X_2,X_3,\ldots$ and formulas are
given by the grammar
$\varphi ::= X=Y+Z \, \mid \, X=YZ \, \mid \, \varphi \wedge \varphi\,$,  where $X,Y,Z$ range over the set of variables.  The
satisfiability problem for $\mathcal{L}$ asks, given a formula
$\varphi$, whether there is an assignment of positive integers to the
variables that satisfies~$\varphi$.  In
Appendix~\ref{append-variant-Diophantine-prop} we show that the
satisfiability problem for $\mathcal{L}$ is undecidable by reduction
from Hilbert's Tenth Problem.

%\begin{proposition} \label{variant-Diophantine-prop} 
%The satisfiability problem for $\mathcal{L}$ is undecidable.
%\end{proposition}

%We prove Proposition~\ref{variant-Diophantine-prop} by  a reduction from Hilbert's Tenth Problem, 
%see Appendix~\ref{append-variant-Diophantine-prop}. 
%We are now ready to prove that:

\begin{theorem}\label{Theorem:Undecidability}
The Pareto Domination Problem is undecidable.
\end{theorem}
\begin{proof}
  Consider the following problem of reaching a single valuation
  in $\RP^{\mathcal{Y}}$:
    given an MPTA
    $\mathcal{A}=\tuple{\locs,\loc_0,L_f,\clocks,\mathcal{Y},\edges, R}$, and
    target valuation~$\gamma \in \RP^{\mathcal{Y}}$, decide whether there is an
    accepting run~$\rho$ of $\mathcal{A}$ such that $\cost(\rho)=\gamma$.

\begin{figure}[t]
\begin{minipage}{0.5\textwidth}
\begin{description}
 
\item[Integer test] $\frac{1}{x_i^{*}} \stackrel{?}{\in}\N$\textbf{:}
  \vspace*{-.4cm}
  \begin{figure}[H]
    \begin{center}
     \scalebox{.85}{\begin{tikzpicture}[->,>=stealth',shorten >=1pt,auto,node distance=1cm, initial text={}]
  
\node [state,draw=none] (0,0) (st)  {\scriptsize{$\{c=0\}$}};
	\node [state] (q1) [right=1.4cm of st] {\scriptsize{$\dot{c}=0$}};
  \node[state] (q2) [right=1.2cm of q1]  {\scriptsize{$\dot{c}=1$}};
  \node[state, draw=none] (en) [right=1.4cm of q2]  {\scriptsize{$\{c=1\}$}};
  
	\draw [draw=magenta, thick,dashed] (1,-.85) rectangle (5.9,1);

\path[->] (st) edge  node [above,near end] {\scriptsize{$r=1~~$}}  node [below,near end] {\scriptsize{$r\leftarrow 0~~$}} (q1);%
	\path[->] (q1) edge [bend left]  node [above,midway] {\scriptsize{$x_i=1$}}  node [below,midway] {\scriptsize{$x_i\leftarrow 0$}} (q2);%
	\path[->] (q2) edge [bend left] node [above,midway] {\scriptsize{$r=1$}}  node [below,midway] {\scriptsize{$r\leftarrow 0$}} (q1);%
		\path[->] (q1) edge [loop above, looseness=5]  node [left,midway] {\scriptsize{\emph{wrap} \,}}  (q1);%
	\path[->] (q2) edge  node [above,near start] {\scriptsize{Inv}}   (en);%
	\path[->] (q2) edge [loop above, looseness=5]  node [right,midway] {\, \scriptsize{\emph{wrap} }}  (q2);%

\end{tikzpicture}}
      \label{fig:test-gadget}
			\end{center}
  \end{figure}
 \vspace*{-1cm}

\item[Decrement] $c \leftarrow c+1-x_i^{*}$\textbf{:}
  \vspace*{-0.4cm}
  \begin{figure}[H]
    \begin{center}
      \scalebox{.85}{\begin{tikzpicture}[->,>=stealth',shorten >=1pt,auto,node distance=1cm, initial text={}]
  \node [state, draw=none] (0,0) (st)  {};
	\node [state] (q1) [right=1.4cm of st] {\scriptsize{$\dot{c}=1$}};
  \node[state] (q2) [right=1.2cm of q1]  {\scriptsize{$\dot{c}=0$}};
  \node[state, draw=none] (en) [right=1.2cm of q2]  {};
  
	\draw [draw=magenta, thick,dashed] (.8,-.65) rectangle (5.7,1);
		
	\path[->] (st) edge  node [above,near end] {\scriptsize{$r=1~~$}}  node [below,near end] {\scriptsize{$r\leftarrow 0~~$}} (q1);%
	\path[->] (q1) edge  node [above,midway] {\scriptsize{$x_i=1$}}  node [below,midway] {\scriptsize{$x_i\leftarrow 0$}} (q2);%
	\path[->] (q1) edge [loop above, looseness=5]  node [left,midway] {\scriptsize{\emph{wrap} \,}}  (q1);%
	\path[->] (q2) edge  node [above,near start] {\scriptsize{Inv}}   (en);%
	\path[->] (q2) edge [loop above, looseness=5]  node [right,midway] {\, \scriptsize{\emph{wrap}}}  (q2);%

\end{tikzpicture}}
      \label{fig:dec-gadget}
    \end{center}
  \end{figure}
	\end{description}
\end{minipage}
\hfill\textcolor{blue!30!white}{\vline}\hfill
\begin{minipage}{0.45\textwidth}

\begin{description}
\item[Quotient] $c \leftarrow c+\frac{x^{*}_i}{x^{*}_j}$\textbf{:}
 \vspace*{-.4cm}
  \begin{figure}[H]
    \begin{center}
      \scalebox{.85}{\begin{tikzpicture}[->,>=stealth',shorten >=1pt,auto,node distance=1cm, initial text={}]
  \node [state, draw=none] (0,0) (st)  {\scriptsize{$\{e=0\}$}};
	\node [state] (q1) [right=1.3cm of st] {\scriptsize{$\dot{e}=0$}};
  \node[state] (q2) [right=1.5cm of q1]  {\scriptsize{$\dot{e}=1$}};
  \node [state] (q3) [below=1.5cm of q2] {\scriptsize{$\dot{c}=0$}};
  \node[state] (q4) [left=1.5cm of q3]  {\scriptsize{$\dot{c}=1$}};
  \node[state, draw=none] (en) [left=1.3cm of q4]  {\scriptsize{$\{e=1\}$}};
  
	\draw [draw=magenta, thick,dashed] (1.05,-3.5) rectangle (5.9,1);
		
	\path[->] (st) edge  node [above,near end] {\scriptsize{$r=1~~$}}  node [below,near end] {\scriptsize{$r\leftarrow 0~~$}} (q1);%
	\path[->] (q1) edge  node [above,midway] {\scriptsize{$x_j=1$}}  node [below,midway] {\scriptsize{$x_j\leftarrow 0$}} (q2);%
	\path[->] (q1) edge [loop above, looseness=5]  node [left,midway] {\scriptsize{\emph{wrap} \,}}  (q1);%
	\path[->] (q2) edge  node [right,midway] {\scriptsize{$r\leftarrow 0$}}  node [left,midway] {\scriptsize{$r=1$}} (q3);%
	\path[->] (q2) edge [loop above, looseness=5]  node [right,midway] {\scriptsize{\, \emph{wrap} }}  (q2);%
	\path[->] (q3) edge  node [above,midway] {\scriptsize{$x_i=1$}}  node [below,midway] {\scriptsize{$x_i\leftarrow 0$}} (q4);%
	\path[->] (q3) edge [loop below, looseness=5]  node [right,midway] {\scriptsize{\, \emph{wrap}}}  (q3);%
	\path[->] (q4) edge  node [right,midway] {\scriptsize{$r\leftarrow 0$}}  node [left,midway] {\scriptsize{$r=1$}} (q1);%
	\path[->] (q4) edge [loop below, looseness=5]  node [left,midway] {\scriptsize{\emph{wrap} \,}}  (q4);%
	\path[->] (q4) edge  node [above,near start] {\scriptsize{Inv}}       (en);%

\end{tikzpicture}}
      \label{fig:dep-gadget}
    \end{center}
  \end{figure}
\end{description}

\end{minipage}
\vspace{-.3cm}
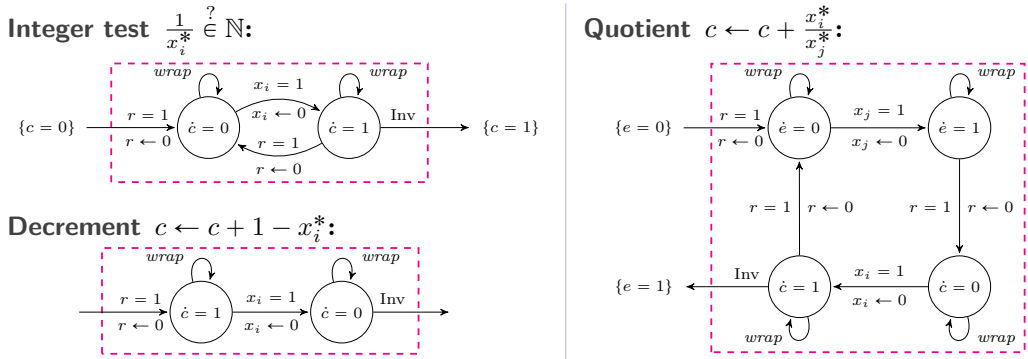
\caption{The \emph{wrap} self-loop denotes a family of $m$ wrapping
  edges, as in~\cite[Fig. 14]{Henzinger98}, where the $j$-th edge has guard~$x_j=1$ and
  resets~$x_j$.  In the quotient gadget, $e$ is a fresh observer, as is $c$ in the integer test.  The
  integer test and quotient gadgets are annotated with predicates
  in curly brackets indicating the initial values
  of observers on entering and
  their target values on exiting the gadget.
  Enforcing these target values through a
    corresponding Pareto constraint guarantees the
  desired behaviour of the gadget.
%In the integer test gadget, $c=0$ is the cost precondition on entering  and~$c=1$ is the
% postcondition on leaving this gadget,
% while $e=0$  and $e=1$ are the
%   cost precondition  and
%    postcondition in the quotient gadget.
}\label{Fig:gadgets}
		\vspace{-.4cm}
\end{figure}

  One can reduce the problem of reaching a given valuation to the Pareto Domination Problem as follows.
  Transform the MPTA $\mathcal{A}$ to an MPTA $\mathcal{A'}$ that has the same locations
  and edges as $\mathcal{A}$ but with two copies of each observer 
  $y\in\mathcal{Y}$, with each copy having the same rate as $y$ in
  each location.  Formally $\mathcal{A'}$ has set of observers
  $\mathcal{Y}'=\{y_1,y_2 : y \in \mathcal{Y}\}$, where $y_1$ is a
  cost variable and $y_2$ is a reward variable.  Then, defining~$\gamma' \in \RP^{\mathcal{Y}'}$ 
  by~$\gamma'(y_1)=\gamma'(y_2)=\gamma(y)$, we have that~$\mathcal{A'}$ has an
  accepting run $\rho'$ such that $\cost(\rho')$ dominates $\gamma'$ just
  in case $A$ has an accepting run $\rho$ such that
  $\cost(\rho)=\gamma$.

  Now we give a reduction from the satisfiability problem for
  $\mathcal{L}$ to the problem of reaching a single valuation.
  Consider an $\mathcal{L}$-formula $\varphi$ over variables
  $X_1,\ldots,X_m$.  We define an MPTA~$\mathcal{A}$ over the set of
  clocks $\clocks=\{x_1,\cdots,x_m,r\}$.  Clock $x_i$ corresponds to
  the variable~$X_i$, for~$i=1,\ldots,m$, while $r$ is a
  \emph{reference clock}.  The reference clock is reset whenever it
  reaches~$1$ and is not otherwise reset---thus it keeps track of
  global time modulo one.  After an initialisation phase the remaining
  clocks $x_1,\ldots,x_m$ are likewise reset in a cyclic fashion,
  whenever they reach~$1$ and not otherwise.  We denote by $x_i^*$ the
  value of clock $x_i$ whenever $r$ is $1$.  During the initialisation
  phase the values~$x_i^*$ are established non-deterministically such
  that $0< x_i^{*} \leq 1$.  The idea is that $\frac{1}{x_i^*}$
  represents the value of variable $X_i$ in $\varphi$;
 in particular,~$x_i^{*}$ is the reciprocal of a positive integer.  For each atomic
  sub-formula in $\varphi$ the automaton $\mathcal{A}$ contains a
  gadget that checks that the guessed valuation satisfies the
  sub-formula.

To present the reduction we first define three primitive gadgets.
  The first ``integer test''  gadget checks that the
initial value $x_i^*$ of clock $x_i$ is a
  reciprocal of a positive integer, by adding
	wrapping edges on all clocks $x_j$ other than $x_i$ to the MPTA from
    Figure~\ref{fig:exam-init}.
% 
%The remaining two gadgets perform certain updates on cost variables in an MPTA. 
% 
The construction of each gadget is such that the precondition $r=0$ holds when control enters
the gadget and the postcondition $r=1 \wedge \bigwedge_{j=1}^m x_j \le 1$ holds on
exiting the gadget.  This last postcondition is abbreviated to Inv in the figures. 
For an observer $c$ and $1\leq i,j\leq m$, we define these three gadgets as in Figure~\ref{Fig:gadgets}.

In the following we show how to compose the three primitive operations
in an MPTA to enforce the atomic constraints in the language
$\mathcal{L}$.  The initialisation automaton below is such that for
$i=1,\ldots,m$ the value $x_i^*$ of clock $x_i$ is such that
$\frac{1}{x_i^*} \in \mathbb{N}$.  Herein the Guess self-loop denotes
a family of $m$ edges, where the $j$-th edge
non-deterministically resets
clock~$x_j$. Note that the incoming edge of the
  integer test gadget enforces $r=1$ such that the initial guesses for
  the clocks $x_i$ satisfy $x_i^* \in [0,1]$. Of these, only
  reciprocals $\frac{1}{x_i^*} \in \mathbb{N}$ pass the subsequent
  series of integer tests.

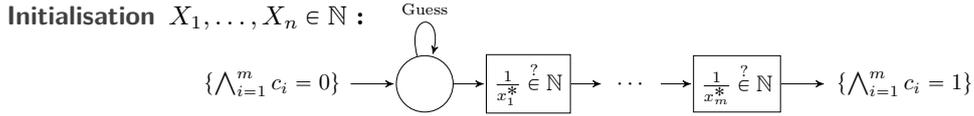
\begin{description}

\item[Initialisation] $X_1,\ldots,X_n \in \mathbb{N}$ \textbf{:} 
 \vspace*{-.9cm}
  \begin{figure}[H]
    \begin{center}
      ~~~~~~~~~~~~~\scalebox{.85}{\begin{tikzpicture}[->,>=stealth',shorten >=1pt,auto,node distance=1cm, initial text={}]
  
\node [state,draw=none] (0,0) (st)  {$\{\bigwedge_{i=1}^m c_i=0\}$};
\node [state] (q0) [right=.7cm of st]  {};
	\node [state,rectangle] (q1)[right=.5cm of q0] {$\frac{1}{x_1^{*}} \stackrel{?}{\in}\N$};
	  \node[state, draw=none] (en) [right=.5cm of q1]  {$\cdots$};
	\node [state,rectangle] (q2)[right=.5cm of en] {$\frac{1}{x_m^{*}} \stackrel{?}{\in}\N$};
  \node[state, draw=none] (enm) [right=.7cm of q2]  {$\{\bigwedge_{i=1}^m c_i=1\}$};

	\path[->] (st) edge (q0);%
        \path[->] (q0) edge [loop above]  node [above,midway] {\scriptsize{Guess}}  (q0);%	
	\path[->] (q0) edge (q1);%
	\path[->] (q1) edge  (en);%
	\path[->] (en) edge  (q2);%
	\path[->] (q2) edge    (enm);%

\end{tikzpicture}}
      \label{fig:init-gadget}
    \end{center}
  \end{figure}
 \vspace*{-1.2cm}
%For each clock variable~$x_i$ there is a copy of
%the gadget shown in Figure~\ref{}. These copies are arranged in sequence
%to form the initialization gadget such that $x^{*}_i$ is the reciprocal
%of some positive integer. 
	\item[Sum] $X_i=X_j+X_k$\textbf{:} According to the encoding of integer
          value $X_n$ as clock value $x_n=\frac 1 {X_n}$, we have to
          enforce
          $\frac{1}{x_i^{*}}=\frac{1}{x_j^{*}}+\frac{1}{x_k^{*}}$,
          which is achieved by the following sequential combination of
          two quotient gadgets.
	
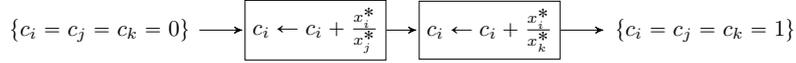
\vspace*{-1.2cm}
	\begin{figure}[H]
 \begin{center}
   \scalebox{.85}{\begin{tikzpicture}[->,>=stealth',shorten >=1pt,auto,node distance=1cm, initial text={}]
 \node [state, draw=none] (0,0) (st)  {$\{c_i=c_j=c_k=0\}$};
	\node [state,rectangle] (q1)[right=.7cm of st] {$c_i \leftarrow c_i+\frac{x_i^{*}}{x^{*}_j}$};
	\node [state,rectangle] (q2)[right=.5cm of q1] {$c_i \leftarrow c_i+\frac{x_i^{*}}{x^{*}_k}$};
	\node [state,draw=none] (en) [right=.7cm of q2] {$\{c_i=c_j=c_k=1\}$};
	
	\path[->] (st) edge   (q1);%
	\path[->] (q1) edge   (q2);%
	\path[->] (q2) edge   (en);%
	
\end{tikzpicture}}
    \label{fig:sum-gadget}
		\end{center}
\end{figure}
	
 \vspace*{-1.4cm}	
\item[Product] $X_i=X_jX_k$\textbf{:} The following gadget enforces $\frac{1}{x_i^{*}}=\frac{1}{x_j^{*}}\cdot \frac{1}{x_k^{*}}$:
  
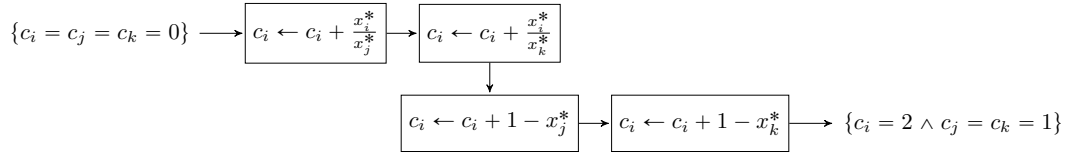
\vspace*{-1.2cm}
  \begin{figure}[H]
    \begin{center}
      \scalebox{.85}{\begin{tikzpicture}[->,>=stealth',shorten >=1pt,auto,node distance=1cm, initial text={}]
  \node [state, draw=none] (0,0) (st)  {$\{c_i=c_j=c_k=0\}$};
	\node [state,rectangle] (q1)[right=.7cm of st] {$c_i \leftarrow c_i+\frac{x_i^{*}}{x^{*}_j}$};
	\node [state,rectangle] (q2)[right=.5cm of q1] {$c_i \leftarrow c_i+\frac{x_i^{*}}{x^{*}_k}$};
	\node [state,rectangle] (q3)[below=.5cm of q2] {$c_i \leftarrow c_i+1-x_j^{*}$};
	\node [state,rectangle] (q4)[right=.5cm of q3] {$c_i \leftarrow c_i+1-x_k^{*}$};
	\node [state,draw=none] (en) [right=.7cm of q4] {$\{c_i=2 \wedge c_j=c_k=1\}$};
	
	\path[->] (st) edge   (q1);%
	\path[->] (q1) edge   (q2);%
	\path[->] (q2) edge   (q3);%
	\path[->] (q3) edge   (q4);%
	\path[->] (q4) edge (en);%
	
\end{tikzpicture}}
      \label{fig:prod-gadget}
    \end{center}
  \end{figure}
	 \vspace*{-2.4cm}
\end{description}

The satisfiability problem for a given $\mathcal{L}$ formula $\varphi$
can now directly be reduced to the problem of reaching a single
valuation $\gamma\in \RP^{\mathcal{Y}}$ by translating each of the
conjuncts of $\varphi$ into the corresponding above MPTA gadget.  The
valuation $\gamma$ encodes the target costs of the respective
gadgets.
\end{proof}

Let us remark that the proof of Theorem~\ref{Theorem:Undecidability}
shows that undecidability of the Pareto Domination Problem already holds in case all
observers  have rates in~$\{0,1\}$. 
Separately we observe that undecidability also holds in the special case that exactly 
one observer  is a cost variable and the others are reward variables, 
and likewise when exactly one observer  is a reward variable and the others are cost variables, when
allowing multiple rates beyond $\{0,1\}$.
The idea is to reduce the problem of reaching a
particular valuation $\gamma \in \RP^{\mathcal{Y}}$ in an MPTA
$\mathcal{A}$ to that of dominating a valuation~$\gamma' \in \RP^{\mathcal{Y}'}$ in a derived MPTA $\mathcal{A'}$ with
set of observers $\mathcal{Y}'=\mathcal{Y}\cup \{y_{\mathrm{sum}}\}$,
where $y_{\mathrm{sum}}$ is a fresh variable.    In
$\mathcal{A'}$ we designate all $y\in \mathcal{Y}$ as cost variables
and $y_{\mathrm{sum}}$ as a reward variable, or vice versa. 
Valuation $\gamma'$ is
specified by $\gamma'(y)=\gamma(y)$ for all $y\in \mathcal{Y}$ and~$\gamma'(y_{\mathrm{sum}}) = \sum_{y\in \mathcal{Y}} \gamma(y)$.
 Automaton
$\mathcal{A'}$ has the same locations, edges, and rate function
as those of $\mathcal{A}$ except that
$R'(y_{\mathrm{sum}}) = \sum_{y\in\mathcal{Y}} R(y)$.

%=========simplex automaton===================
\section{The Simplex Automaton}\label{sec:simplex-automaton}
This section introduces the basic construction from which we derive
our positive decidability results and complexity upper bounds.

Let
$\mathcal{A}=\langle L,\loc_0, L_f,
\mathcal{X},\mathcal{Y},E,R\rangle$ be an MPTA.  For a sequence of
edges $e_1,\ldots,e_m \in E$, define
$\mathit{Runs}(e_1,\ldots,e_m) \subseteq \RP^{m}$ to be the collection
of sequences of timestamps $(t_1,\ldots,t_m)\in \RP^{m}$ such that
$\mathcal{A}$ has a run
$ \rho = (\ell_0,\nu_0,t_0) \stackrel{e_1}{\longrightarrow}
(\ell_1,\nu_1,t_1)\stackrel{e_2}{\longrightarrow} \ldots
\stackrel{e_m}{\longrightarrow} (\ell_m,\nu_m,t_m)$.  Recalling
that by convention $t_0=0$ and $\nu_0=\boldsymbol{0}$, once the edges
$e_1,\ldots,e_m$ have been fixed then the run~$\rho$ is determined
solely by the timestamps $t_1,\ldots,t_m$.  When the sequence of edges~$e_1,\ldots,e_m$ is understood, we call such a sequence of timestamps
a run.  

\begin{proposition}
  $\mathit{Runs}(e_1,\ldots,e_m) \subseteq \RP^{m}$ is
  defined by a conjunction of difference constraints.
\label{prop:diff}
\end{proposition}

The proof of Proposition~\ref{prop:diff} is in
Appendix~\ref{append-prop-diff}.

\begin{proposition}
  $\mathit{Runs}(e_1,\ldots,e_m)$ is equal to the convex hull
  of the set of its integer points.
  \label{prop:hull}
\end{proposition}
\begin{proof}
  Fix a positive integer $M$.  From Proposition~\ref{prop:diff} it
  immediately follows that the set
  $\mathit{Runs}(e_1,\ldots,e_m) \cap [0,M]^m$ can be written as a
  conjunction of closed difference constraints~$A\boldsymbol{t} \leq \boldsymbol{b}$, where $A$ is an integer
  matrix, $\boldsymbol{t}$ the vector of time-stamps $t_1 \ldots t_m$,
  and $\boldsymbol{b}$ an integer vector.  Given this, it follows that
  ${\mathit{Runs}(e_1,\ldots,e_m)}\cap[0,M]^{m}$, being a closed and
  bounded polygon, is the convex hull of its vertices.  Moreover each
  vertex is an integer point since the matrix~$A$ here, being by
  Proposition~\ref{prop:diff} the incidence matrix of a balanced
  signed graph with half edges, is totally
  unimodular~\cite[Proposition 8A.5]{Zaslavsky82}.
\end{proof}

Proposition~\ref{prop:carathedory} shows that for Pareto reachability on
an MPTA $\mathcal{A}$ with $|\mathcal{Y}| = d$ observers, 
it suffices to look at $d+1$-simplices of integer runs.

\begin{proposition}
  For any run $\rho$ of $\mathcal{A}$ there exists a set of at most $d+1$
  integer-time runs $S$, all over the same sequence of edges as
  $\rho$, such that $\cost(\rho)$ lies in the convex hull of
  $\cost(S)$.
  \label{prop:carathedory}
\end{proposition}
\begin{proof}
  Let $\rho$ be a run of $\mathcal{A}$ over an edge-sequence $e_1,\ldots,e_m$
  with time stamps $t_0,\ldots,t_m$, given by $
 \rho = (\ell_0,\nu_0,t_0) \stackrel{e_1}{\longrightarrow}
          (\ell_1,\nu_1,t_1)\stackrel{e_2}{\longrightarrow} \ldots
          \stackrel{e_m}{\longrightarrow} (\ell_m,\nu_m,t_m)$.  By
  Proposition~\ref{prop:hull}, $(t_1,\ldots,t_m)$ lies in the convex
  hull of the set $I$ of integer points in
  ${\mathit{Runs}(e_1,\ldots,e_m)}$.

  Since the map
  $cost:{\mathit{Runs}(e_1,\ldots,e_m)} \rightarrow
  \mathbb{R}^d$ is linear we have that $\cost(\rho)$ lies in
  the convex hull of $\cost(I)$.  Moreover by
  Carath\'{e}odory's Theorem there exists a subset $S\subseteq I$ of
  cardinality at most $d+1$ such that $\cost(\rho)$ lies in
  the convex hull of $\cost(S)$.
\end{proof}

We now exploit
Proposition~\ref{prop:carathedory} by introducing the so-called
\emph{simplex automaton} $\mathcal{S(A)}$, which is a
monotone VASS obtained from a given MPTA $\mathcal{A}$.  The automaton~$\mathcal{S(A)}$ 
generates~$(d+1)$-tuples of integer-time runs of $\mathcal{A}$,
such that each run in the tuple executes the same sequence of edges in
$\mathcal{A}$ and the runs differ only in the times at which the edges are
taken.  The basic component underlying the definition of the simplex
automaton is the \emph{integer-time automaton} $\mathcal{Z(A)}$.  This
automaton is a monotone VASS that generates the integer-time runs
of $\mathcal{A}$, using its counters to keep track of the running cost for each
observer.

The definition of $\mathcal{Z(A)}$ is as follows.  Let
$\mathcal{A}=\tuple{L,\loc_0,L_f,\mathcal{X},\mathcal{Y},E,R}$ be an MPTA.  Let also
$M_{\mathcal{X}}\in \N$ be a positive constant greater than the maximum
clock constant in $\mathcal{A}$.  We define a
monotone VASS $\mathcal{Z(A)} = \tuple{d,Q,q_0,Q_f,E,\Delta}$, in which
the dimension $d = |\mathcal{Y}|$, the set
of states is $Q= L \times \{0,1,\ldots,M_{\mathcal{X}}\}^\mathcal{X}$, the initial
state is $q_0=(\ell_0,\boldsymbol{0})$, the set of accepting states is
$Q_f = L_f\times \{0,1,\ldots,M_{\mathcal{X}}\}^{\mathcal{X}}$, the set of labels is
$E$ (i.e., the set of edges of the MPTA), and the transition relation
$\Delta \subseteq Q \times \mathbb{N}^d \times E \times Q$ includes a
transition $((\ell,\nu), t\cdot R(\ell),e,(\ell',\nu'))$ for every
$t\in \{0,1,\ldots,M_{\mathcal{X}}\}$ and edge $e=(\ell,\varphi,\lambda,\ell')$ in
$\mathcal{A}$ s.t. $\nu\oplus t \models \varphi$ and
$\nu'=(\nu\oplus t)[\lambda\leftarrow 0]$.  Here
$(\nu\oplus t)(x)=\min(\nu(x)+t,M_{\mathcal{X}})$ for all $x\in \mathcal{X}$. We then have:

\begin{proposition}
  Given a valuation $\gamma \in \RP^{\mathcal{Y}}$, there exists an integer-time accepting run
  $\rho$ of $\mathcal{A}$ with $\cost(\rho)=\gamma$
  if and only if $\gamma \in \mathrm{Reach}_{\mathcal{Z(A)}}$.
\label{prop:integer}
\end{proposition}

The simplex automaton $\mathcal{S(A)}$ is built by taking $d+1$
copies of $\mathcal{Z(A)}=\tuple{d,Q,q_0,Q_f,E,\Delta}$
that synchronize on transition labels.  Formally,
$\mathcal{S(A)} = \tuple{d(d+1),Q^{d+1},\boldsymbol{q}_0,Q_f^{d(d+1)},E,\boldsymbol{\Delta}}$, where 
$\boldsymbol{q}_0=(q_0,\ldots,q_0)$ and
$\boldsymbol{\Delta} \subseteq Q^{d+1} \times \Z^{d(d+1)} \times E \times Q^{d+1}$ comprises those 
 tuples $((q_1,\ldots,q_{d+1}),(\boldsymbol{v}_1,\ldots,\boldsymbol{v}_{d+1}),e,(q_1',\ldots,q_{d+1}'))$ 
	s.t. $(q_i,\boldsymbol{v}_i,e,q_i') \in \Delta$ for all $i\in\{1,\ldots,d+1\}$.

From Propositions~\ref{prop:carathedory} and~\ref{prop:integer} we
have:
\begin{proposition}
  Given $\gamma \in \RP^{\mathcal{Y}}$, there exists an accepting run
  $\rho$ of $\mathcal{A}$ with $\mathrm{cost}(\rho)=\gamma$ if and only if there
  exists
  $(\gamma_1,\ldots,\gamma_{d+1}) \in
  \mathrm{Reach}_{\mathcal{S(A)}}$ with
 $\gamma$  in the convex hull of $\{\gamma_1,\ldots,\gamma_{d+1}\}$.
\label{prop:convex-hull}
\end{proposition}

We now introduce the following
``master system'' of bilinear inequalities
that expresses whether $\gamma \preccurlyeq \cost(\rho)$ for some
accepting run $\rho$ of $\mathcal{A}$. 
\begin{equation}
\begin{array}{rclrcl}
\gamma &\preccurlyeq & \lambda_1\gamma_1 + \cdots + \lambda_{d+1}\gamma_{d+1} \qquad\qquad &
1 & = & \lambda_1 + \cdots + \lambda_{d+1} \\
\multicolumn{3}{l}{(\gamma_1,\ldots,\gamma_{d+1}) \in \mathrm{Reach}_{\mathcal{S(A)}}} & 
0 & \leq & \lambda_1,\ldots,\lambda_{d+1}
\end{array}
\label{eq:master-system}
\end{equation}

The system has real variables $\lambda_1,\ldots,\lambda_{d+1} \in \RP^{\mathcal{Y}}$
and integer variables $\gamma_1,\ldots,\gamma_{d+1} \in \mathbb{N}^{\mathcal{Y}}$. 
The key property of the master system  
is stated in the following Proposition~\ref{prop:master}, 
which follows immediately from Proposition~\ref{prop:convex-hull}.
%\begin{equation}
%\gamma \preccurlyeq  \lambda_1\gamma_1 + \cdots + \lambda_{d+1}\gamma_{d+1},
%1  = \lambda_1 + \cdots + \lambda_{d+1},
%0  \leq  \lambda_1,\ldots,\lambda_{d+1},
%\gamma_1,\ldots,\gamma_{d+1} \in \mathrm{Reach}_{\mathcal{S(A)}}
%\label{eq:master-system}
%\end{equation}

%By Proposition~\ref{prop:convex-hull} we have:
\begin{proposition}
  Given a valuation $\gamma\in\RP^{\mathcal{Y}}$ there is an accepting
  run $\rho$ of $\mathcal{A}$ such that~$\gamma\preccurlyeq\cost(\rho)$ if and
  only if the system of inequalities (\ref{eq:master-system}) has a
  solution.
\label{prop:master}
\end{proposition}

Given Proposition~\ref{prop:master}, the results of
Section~\ref{sec:undecidability} imply that satisfiability of the
master system~(\ref{eq:master-system}) is not decidable in general.
In the rest of the paper we pursue different approaches to showing
decidability of restrictions and variants of the Pareto Domination Problem by
solving appropriately restricted versions of (\ref{eq:master-system}).

%==========decidability============
\section{Pareto Domination Problem with Pure Constraints} \label{sec:pure-const}
In this section we show that the Pareto Domination Problem is decidable in polynomial space
for the class of MPTA in which the observers are all costs. We prove this complexity upper
bound by exhibiting for such an MPTA~$\mathcal{A}$ and
target~$\gamma\in\RP^{\mathcal{Y}}$ a positive integer $M$, whose
bit-length is polynomial in the size of $\mathcal{A}$ and $\gamma$,
such that there exists a run $\rho$ of $\mathcal{A}$ reaching the
target location with $\gamma \preccurlyeq \cost(\rho)$ iff there
exists such a run of granularity $\frac{1}{M_1}$ for some $M_1\leq M$.
To show this we rewrite the bilinear system of inequalities
(\ref{eq:master-system}) into an equisatisfiable disjunction of linear
systems of inequalities. We thus obtain a bound on the bit-length of
any satisfying assignment of~(\ref{eq:master-system}) from which we
obtain the above granularity bound. A
similar bound in case of all reward variables is obtained in~\ref{append-all-reward}.

%\subsection{Semi-linear Decomposition of
%$\mathrm{Reach}_{\mathcal{S(A)}}$}

Consider an MPTA
$\mathcal{A}=\tuple{L,\ell_0,L_f,\mathcal{X},\mathcal{Y},E,R}$. Recall
that the reachability set $\mathrm{Reach}_{\mathcal{S(A)}}$ can be
written as a union of linear sets $S(\boldsymbol{v}_i,P_i)$, $i\in I$.
More precisely, let $M_\mathcal{Y}$ be the maximum rate occurring in
the rate function $R$ of the given MPTA $\mathcal{A}$.  We then have
the following, see Appendix~\ref{append-prop-SL-decomp} for the proof.
\begin{proposition}
  The set $\mathrm{Reach}_{\mathcal{S(A)}}$ can be written as a
  finite union of linear sets
  $\bigcup_{i\in I} S(\boldsymbol{v}_i,P_i)$ such that for each $i\in I$
  the base vectors $\boldsymbol{v}_i$ and period vectors in $P_i$ have
  entries of magnitude bounded by 
  $\mathit{poly}(d,|L|,M_{\mathcal{Y}},M_{\mathcal{X}})^{d(d+1)|\mathcal{X}|}$.
  \label{prop:SL-decomp}
\end{proposition}

%\subsection{All Cost Variables}

Suppose that the set of observers $\mathcal{Y}$ with $|\mathcal{Y}|=d$ is comprised exclusively of
cost variables.  We will apply Proposition~\ref{prop:SL-decomp} to analyse
the Pareto Domination Problem.  The key observation is that in this case we
can equivalently rewrite the bilinear system (\ref{eq:master-system})
as a disjunction of linear systems of inequalities.

As a first step we can rewrite the constraint
$(\gamma_1,\ldots,\gamma_{d+1}) \in \mathrm{Reach}_{\mathcal{S(A)}}$
in (\ref{eq:master-system}) as a disjunction of constraints
$(\gamma_1,\ldots,\gamma_{d+1}) \in S(\boldsymbol{v}_i,P_i)$, for
$i\in I$.  But since the period vectors in~$P_i$ are non-negative we
can further observe that in order to satisfy the upper bound
constraints on cost variables, the optimal choice of
$(\gamma_1,\ldots,\gamma_{d+1}) \in S(\boldsymbol{v}_i,P_i)$ is the
base vector $\boldsymbol{v}_i$.  Thus we can treat
$\gamma_1,\ldots,\gamma_{d+1}$ as a constant in
(\ref{eq:master-system}).

Thus we rewrite (\ref{eq:master-system}) as a finite disjunction of
systems of linear inequalities---one such system for each $i\in I$.
For a given $i\in I$ let
$\boldsymbol{v}_i = (\gamma^{(i)}_1,\ldots,\gamma^{(i)}_{d+1})$ be the base vector
of the linear set $S(\boldsymbol{v}_i,P_i)$.  The corresponding system
of inequalities specialising (\ref{eq:master-system}) is
\begin{equation}
\gamma \preccurlyeq \lambda_1\gamma^{(i)}_1 + \ldots + \lambda_{d+1}\gamma^{(i)}_{d+1},\; 1  =  \lambda_1 + \cdots + \lambda_{d+1},\;  0  \leq  \lambda_1,\ldots,\lambda_{d+1}
\label{eq:LP}
\end{equation}

Recall that if a set of linear inequalities
$A\boldsymbol{x}\geq \boldsymbol{a}$, $B\boldsymbol{x}>\boldsymbol{b}$
is feasible then it is satisfied by some
$\boldsymbol{x}\in\mathbb{Q}^n$ of bit-length $\mathit{poly}(n,b)$,
where $b$ is the total bit-length of the entries
of~$A$,~$B$,~$\boldsymbol{a}$, and $\boldsymbol{b}$.  Applying this
bound and Proposition~\ref{prop:SL-decomp} we see that a solution of
(\ref{eq:LP}) can be written in the form
$\lambda_1=\frac{p_1}{g},\ldots,\lambda_{d+1}=\frac{p_{d+1}}{g}$ for
integers $p_1,\ldots,p_{d+1},g$ of bit-length at most
$\mathit{poly}(d,|\mathcal{X}|,|L|,\log(M_{\mathcal{Y}}),\log(M_{\mathcal{X}}))$.
This entails that the cost vector
$\lambda_1\gamma^{(i)}_1 + \ldots + \lambda_{d+1}\gamma^{(i)}_{d+1}$
arises from a run of $\mathcal{A}$ with granularity $\frac{1}{g}$,
thus indirectly addressing the open problem stated in~\cite[Section 8]{Larsen08}
 on the granularity of optimal runs in MPTA.

Together with Proposition~\ref{prop:SL-decomp}, this yields
PSPACE-membership for the Pareto Domination Problem. As reachability in timed automata
is already PSPACE-hard~\cite{AlurDill94} we have:
%PSPACE-hardness follows from a
%reduction of the reachability problem of timed automata known to be
%PSPACE-complete~\cite{AlurDill94}, by extending the arguments
%in~\cite[Section 4.4]{Bouyer07} to multiple observers. This leads to
%Theorem~\ref{thm:pure-constraints} that strengthens the results
%in~\cite[Theorem 1 and Corollary 1]{Larsen08}:  
% We already said this in the introudcton, so I don't think we need to say it again
\begin{theorem}\label{thm:pure-constraints} The Pareto Domination Problem with pure constraints  is PSPACE-complete.
\end{theorem}

\section{Pareto Domination Problem with Three Mixed Observers} \label{sec:three-cost-variables}
In this section we consider the Pareto Domination Problem for MPTA
with three observers.  In the case of three cost variables or
three reward variables the results of Section~\ref{sec:pure-const}
apply.  Below we show decidability for two cost variables and one
reward variable. The similar case of two reward variables and one cost
variable is handled in Appendix~\ref{append-three-var}.

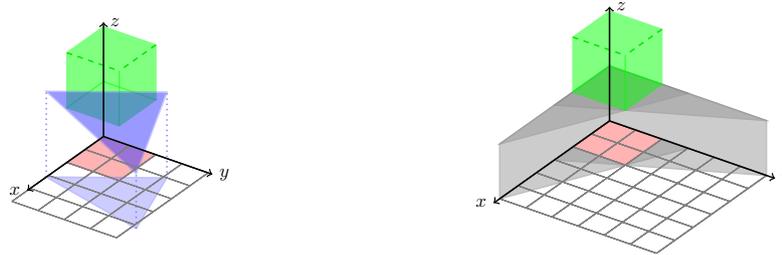
\begin{figure}[t]
\centering

\begin{minipage}{0.4\textwidth}
\begin{center}
   \scalebox{.7}{\tdplotsetmaincoords{60}{125}
\begin{tikzpicture}
	[tdplot_main_coords,
		grid/.style={very thin,gray},
		axis/.style={->, thick},
		cube/.style={opacity=.5, thick,green,fill=green}]
	
	%draw the bottom of the cube
	\draw[cube] (0,0,1.2) -- (1.2,0,1.2) -- (1.2,1.2,1.2) -- (0,1.2,1.2) -- cycle;
	
	%draw the back-right of the cube
	\draw[cube] (0,0,1.2) -- (0,0,2.4) -- (0,1.2,2.4) -- (0,1.2,1.2) -- cycle;

	%draw the back-left of the cube
	\draw[cube] (0,0,1.2) -- (0,0,2.4) -- (1.2,0,2.4) -- (1.2,0,1.2) -- cycle;

  \draw[cube]  (1.2,1.2,1.2) --(1.2,1.2,2.4);
	\draw[green!90!black, thick, dashed]  (0,1.2,2.4) --(1.2,1.2,2.4);
  \draw[green!90!black, thick, dashed]  (1.2,0,2.4) --(1.2,1.2,2.4);
	
   %draw the IMAGE of the cube
	\draw[thick,red!50!white, fill=red!30!white] (0,0,0) -- (1.2,0,0) -- (1.2,1.2,0) -- (0,1.2,0) -- cycle;
	
	 %draw the IMAGE of the Traingle
	\draw[thick,opacity=.5,blue!50!white, fill=blue!40!white] (1.875,0,0) -- (.6,1.875,0) -- (2.5,2.5,0) -- cycle;
	
	 %draw the the Traingle
	\draw[thick,opacity=.5,blue!50!white, fill=blue!80!white] (1.875,0,1.875) -- (.6,1.875,1.875) -- (2.5,2.5,1.2) -- cycle;

	\draw[blue!50!white, thick, dotted]  (1.875,0,0) --(1.875,0,1.875);
  \draw[blue!50!white, thick, dotted]  (.6,1.875,0) --(.6,1.875,1.875);
	\draw[blue!50!white, thick, dotted]  (2.5,2.5,0) --(2.5,2.5,1.2);

%draw a grid in the x-y plane
	\foreach \x in {0,.6,...,3.33}
		\foreach \y in {0,.6,...,2.5}
		{
			\draw[grid] (\x,0) -- (\x,2.4);
			\draw[grid] (0,\y) -- (3,\y);
		}			

	%draw the axes
	\draw[axis] (0,0,0) -- (2.5,0,0) node[left]{$x$};
	\draw[axis] (0,0,0) -- (0,2.5,0) node[anchor=west]{$y$};
	\draw[axis] (0,0,0) -- (0,0,2.5) node[anchor=west]{$z$};

% \draw [fill=blue!50!white,blue!50!white](2.25,0,2.25) circle (2pt) node [above,blue]{$\boldsymbol{p}$};
% \draw [fill=blue!50!white,blue!50!white](.75,2.25,2.25) circle (2pt) node [above,blue]{$\boldsymbol{r}$};
% \draw [fill=blue!50!white,blue!50!white](3,3,1.5) circle (2pt) node [right,blue]{$\boldsymbol{q}$};

% \draw [fill=red,red](1.5,0,0) circle (2pt) node [above,red]{$\boldsymbol{a}$};
% \draw [fill=red,red](0,1.5,0) circle (2pt) node [above,red]{$\boldsymbol{b}$};
% \draw [fill=red,red](1.5,1.5,0) circle (2pt) node [below,red]{$\boldsymbol{c}$};
\end{tikzpicture}}
		\end{center}
\end{minipage}
\quad\quad
\begin{minipage}{0.45\textwidth}
\centering
 \scalebox{.7}{\tdplotsetmaincoords{60}{125}
\begin{tikzpicture}
	[tdplot_main_coords,
		grid/.style={very thin,gray},
		axis/.style={->, thick},
		cube/.style={opacity=.5, thick,green,fill=green}]
	
		\coordinate (C) at (1.2,1.2,0); 
	\coordinate (Ca) at (1.8,0,0);
	\coordinate (Cb) at (0,1.8,0);
  \coordinate (orig) at (0,0,0);
	\coordinate (xax) at (3.33,0,0);
	\coordinate (yax) at (0,3.33,0);
	\coordinate (G1) at (intersection of C--Ca and orig--yax);
  \coordinate (G2) at (intersection of C--Cb and orig--xax);
	\coordinate (X0) at (.6,.6,0); 
  \coordinate (Y0) at (1.8,1.8,0); 
	
	 %draw the IMAGE of F
	\draw[thick,opacity=.5,black!50!white, fill=black!40!white] (G1) -- (orig) -- (Ca) -- cycle;
	\draw[thick,opacity=.5,black!50!white, fill=black!40!white] (G2) -- (orig) -- (Cb) -- cycle;
	\draw[thick,opacity=.5,black!50!white, fill=black!40!white] (3.6,0,1.2) -- (0,0,1.2) -- (0,1.8,1.2) -- cycle;
	\draw[thick,opacity=.5,black!50!white, fill=black!40!white] (0,3.6,1.2) -- (0,0,1.2) -- (1.8,0,1.2) -- cycle;
	\draw[thick,opacity=.5,black!50!white, fill=black!40!white] (3.6,0,0) -- (3.6,0,1.2) -- (0,0,1.2) -- (orig)--cycle;
  \draw[thick,opacity=.5,black!50!white, fill=black!40!white] (0,3.6,0) -- (0,3.6,1.2) -- (0,0,1.2) -- (orig)--cycle;
	
		%draw the bottom of the cube
	\draw[cube] (0,0,1.2) -- (1.2,0,1.2) -- (1.2,1.2,1.2) -- (0,1.2,1.2) -- cycle;
	
	%draw the back-right of the cube
	\draw[cube] (0,0,1.2) -- (0,0,2.4) -- (0,1.2,2.4) -- (0,1.2,1.2) -- cycle;

	%draw the back-left of the cube
	\draw[cube] (0,0,1.2) -- (0,0,2.4) -- (1.2,0,2.4) -- (1.2,0,1.2) -- cycle;

  \draw[cube]  (1.2,1.2,1.2) --(1.2,1.2,2.4);
	\draw[green!90!black, thick, dashed]  (0,1.2,2.4) --(1.2,1.2,2.4);
  \draw[green!90!black, thick, dashed]  (1.2,0,2.4) --(1.2,1.2,2.4);
	
\draw[thick,red!50!white, fill=red!30!white] (0,0,0) -- (1.2,0,0) -- (1.2,1.2,0) -- (0,1.2,0) -- cycle;

%draw a grid in the x-y plane
	\foreach \x in {0,.6,...,3.8}
		\foreach \y in {0,.6,...,3.8}
		{
			\draw[grid] (\x,0) -- (\x,3.6);
			\draw[grid] (0,\y) -- (3.6,\y);
		}			

	%draw the axes
	\draw[axis] (0,0,0) -- (3.8,0,0) node[left]{$x$};
	\draw[axis] (0,0,0) -- (0,3.8,0) node[anchor=west]{$y$};
	\draw[axis] (0,0,0) -- (0,0,2.5) node[anchor=west]{$z$};
	%
	%\draw [fill=black,black](X0) circle (2pt) node [left,black]{$\pi(\boldsymbol{x})$};
%\draw [fill=black,black](Y0) circle (2pt) node [right,black]{$\pi(\boldsymbol{y})$};
%\draw (X0) -- (Y0);

\end{tikzpicture}}

\end{minipage}
\vspace{-.3cm}
\caption{The target~$T$ is the green rectangular region and
		the blue region is~$S$. The pink region is $\pi(T)$ and the light 
		blue region $\pi(S)$. The grey region~$F$ is described in equation~(\ref{eq:def-set-F}).}
\label{fig:three-Dim}
\vspace{-.5cm}
\end{figure}

Consider an instance of the Pareto Domination Problem given by an MPTA
$\mathcal{A}$ with~$|\mathcal{Y}|=3$ observers, and
a target vector $\gamma\in\RP^{\mathcal{Y}}$.  Our starting point is
again Proposition~\ref{prop:master}.  To apply this proposition the
idea is to eliminate the quantifiers over the real variables (the
$\lambda_i$) in the system of equations (\ref{eq:master-system}) and
thereby obtain a formula that lies in a decidable fragment of
arithmetic (namely disjunctions of constraints of the form
considered in Theorem~\ref{thm:segal}).

To explain this quantifier-elimination step in more detail, let us
identify $\RP^{\mathcal{Y}}$ with $\RP^3$.  Denote by
$T\subseteq \RP^3$ the set of valuations that dominate a given fixed
valuation $\gamma \in \RP^3$.  We can write
$T = \{ (x,y,z) \in \RP^3 : x \leq a \wedge y \leq b \wedge z \geq c
\} \, , $ where $a,b,c$ are non-negative integer constants (see the
left-hand side of Figure~\ref{fig:three-Dim}).  We seek a
quantifier-free formula of arithmetic that expresses that  $T$ meets a 4-simplex~$S\subseteq \RP^3$
given by the convex hull of~$\{\gamma_1,\ldots,\gamma_4\}$,
where
$(\gamma_1,\ldots,\gamma_4) \in \mathrm{Reach}_{S(\mathcal{A})}$.
However, since $T$ is unbounded, it is clear that $T$ meets a given
4-simplex $S$ just in case it meets a face of $S$ (which is a
3-simplex).  Thus it will suffice to write a quantifier-free formula
of arithmetic $\varphi_T$ expressing that a 3-simplex in~$\RP^3$ meets~$T$.  Such a formula has nine free variables---one for each of the
coordinates of the three vertices of $S$.  We describe $\varphi_T$ in
the remainder of this section.

%While constructing such a formula is straightforward in principle, care is
%required to remain in a fragment of arithmetic that is known to be
%decidable.  In particular we require that when the formula is
%presented in disjunctive normal form, each term has the form
%identified in Theorem~\ref{thm:segal}, i.e., mentioning at most one
%quadratic polynomial.  
%Since all objectives~$T$ considered in this
%section are unbounded, 

%Consider a reachability target $T \subseteq \RP^3$ given by two
%upper-bound constraints and one lower-bound constraint.
%  Write $T = \{ (x,y,z) \in \RP^3 : x \leq a \wedge y \leq b \wedge z \geq
%    c \} \, , $ where $a,b,c$ are non-negative integer constants (see
%  Figure~\ref{fig:three-Dim} on the left).  

%\begin{figure}[t]
%\centering

%\begin{minipage}{0.4\textwidth}
%\begin{center}
 %  \scalebox{.52}{\input{Figures/threeDimTarg.tex}}
%		\end{center}
%\end{minipage}
%\quad\quad
%\begin{minipage}{0.45\textwidth}
%\centering
% \scalebox{.52}{\input{Figures/FixedPoint3Dtarg.tex}}

%\end{minipage}
%\caption{\textcolor{red}{The target~$T$ is the green rectangular region and
%		the blue region is~$S$. The pink region is $\pi(T)$ and the light 
%		blue region $\pi(S)$. The grey region~$F$ is described in equation~(\ref{eq:def-set-F}).}}
%\label{fig:three-Dim}
%\vspace{-.4cm}
%\end{figure}

%Furthermore denote the vertices of~$\pi(T)$ by
%$\boldsymbol{a}=(a,0)$, $\boldsymbol{b}=(0,b)$ and
%$\boldsymbol{c}=(a,b)$, see Figure~\ref{fig:three-Dim-Targ}.

It is geometrically clear that $S$ intersects $T$ iff either $S$
lies inside $T$, the boundary of $S$ meets $T$, or the boundary of $T$
meets $S$.  More specifically we have the following proposition, whose
proof is given in Appendix~\ref{append-prop-case-split}.
\begin{proposition}
Let $S\subseteq\RP^3$ be a 3-simplex.  Then
$T\cap S$ is nonempty if and only if at least one of the following holds:
(a) Some vertex of $S$ lies in $T$; (b) Some bounding edge of $S$ intersects either the face of $T$
  supported by the plane $x=a$ or the face of $T$ supported by the
  plane $y=b$; (c) The bounding edge of $T$ supported by the line $x=a\cap y=b$ intersects $S$.
\label{prop:case-split}
\end{proposition}

The following definition and proposition are key to expressing
intersections of the form identified in Case (c) of
Proposition~\ref{prop:case-split} in terms of quadratic constraints.
The idea is to identify a bounded region $F\subseteq \RP^3$ such that
in Case (c) one of the vertices of $S$ lies in $F$.  The proof of
Proposition~\ref{prop-3dim-fixed-a-point} can be found in
Appendix~\ref{append-fixed-a-point}.

Define a region $F\subseteq \RP^3$ (depicted as the grey-shaded region
on the right of Figure~\ref{fig:three-Dim}) by:
\begin{gather}
F=\{(x,y,z) \in \RP^3 \mid z<c\wedge (x+ay\leq a(b+1) \vee
y+bx\leq b(a+1))\}.                                                             \label{eq:def-set-F}
\end{gather}
Then we have:
\begin{proposition}\label{prop-3dim-fixed-a-point}
  Let $S\subseteq\RP^3$ be a 3-simplex such that $S\cap T$ is
  non-empty but none of the bounding edges of $S$ meets $T$.  Then some
  vertex of $S$ lies in $F$.
\end{proposition}
%\begin{proof}
%  Since $S\cap T\neq\emptyset$, we have
%  $\pi(S)\cap\pi(T)\neq\emptyset$.  Hence there are vertices
%  $\boldsymbol{x},\boldsymbol{y}$ of $S$ such that the edge
%  $\pi(\boldsymbol{x})\pi(\boldsymbol{y})$ meets $\pi(T)$.  By
%  Proposition~\ref{prop:decomposition-of-F} we have either
%  that one of $\pi(\boldsymbol{x})$ and $\pi(\boldsymbol{y})$ lies in
%  $\pi(T)$ or that both $\pi(\boldsymbol{x})$ and %
%	$\pi(\boldsymbol{y})$ lie in $\pi(F)$.  
%
%  Suppose $\pi(\boldsymbol{x}) \in \pi(T)$.  Since the edge
%  $\boldsymbol{x}\boldsymbol{y}$ is assumed not to meet $T$ we must
%  have that $x_3<c$ and hence $\boldsymbol{x}\in F$.  Likewise the
%  assumption that $\pi(\boldsymbol{y}) \in \pi(T)$ yields
%  $\boldsymbol{y}\in F$.  Finally, if both $\pi(\boldsymbol{x})$ and
%  $\pi(\boldsymbol{y})$ lie in $\pi(F)$ then the assumption that
%  $\boldsymbol{x}\boldsymbol{y}$ does not meet $T$ implies that either
%  $x_3< c$ or $y_3<c$.  Hence $\boldsymbol{x}\in F$ or
%  $\boldsymbol{y}\in F$.
%\end{proof}

Denote by $\pi:\R^3\rightarrow \R^2$ the projection of $\R^3$ onto the
$xy$-plane, where $\pi(x,y,z)=(x,y)$ for all $x,y,z\in\R$.  Write
$\pi(T)$ and $\pi(S)$ for the respective images of $T$ and $S$ under
$\pi$.

We write separate formulas
$\varphi_T^{(1)},\varphi_T^{(2)},\varphi_T^{(3)}$, respectively
expressing the three necessary and sufficient conditions for $T\cap S$
to be nonempty, as identified in Proposition~\ref{prop:case-split}.  These
are formulas of arithmetic whose free variables denote the coordinates
of the three vertices of $S$.

{\bf Some vertex of $S$ lies in $T$.} 
\label{subsec:phi1}
Denote the vertices of $S$ by
$\boldsymbol{p},\boldsymbol{q},\boldsymbol{r}$.  Formula
$\varphi_T^{(1)}$ expresses that $\boldsymbol{p}\in T$ or
$\boldsymbol{q}\in T$ or $\boldsymbol{r}\in T$.  This is clearly a
formula of linear arithmetic.

{\bf Some bounding edge of $S$ meets a face of $T$.}
\label{subsec:phi2}
It is straightforward to obtain $\varphi_T^{(2)}$ given a formula
$\psi$ expressing that an arbitrary line segment
${\boldsymbol{x}\boldsymbol{y}}$ in $\RP^3$ meets a given
fixed face of $T$.  We outline such a formula in the rest of this
sub-section. For concreteness we consider the face of $T$ supported by
the plane $x=a$, which maps under $\pi$ to the line segment
$L=\{(a,y) : 0 \leq y \leq b\}$.
Formula $\psi$ has six
free variables, respectively denoting the coordinates of
$\boldsymbol{x}=(x_1,x_2,x_3)$ and $\boldsymbol{y}=(y_1,y_2,y_3)$.

Formula $\psi$ is a conjunction of two parts.  The first part
expresses that ${\pi(\boldsymbol{x})\pi(\boldsymbol{y})}$
meets~$L$.  Since the complement
of $\pi(F)$ is a convex region in $\RP^2$ that excludes $\pi(T)$ we
have that either $\pi(\boldsymbol{x}) \in \pi(F)$ or
$\pi(\boldsymbol{y})\in \pi(F)$.  Moreover since $\pi(F)$ contains
finitely many integer points, we can write separate sub-formulas
expressing that ${\pi(\boldsymbol{x})\pi(\boldsymbol{y})}$
meets $L$ for each fixed value
of $\pi(\boldsymbol{x}) \in \pi(F)$ and each fixed value of
$\pi(\boldsymbol{y}) \in \pi(F)$. Each of these sub-formulas can then be written
in linear arithmetic, see Appendix~\ref{append-geometry}.

Suppose now that ${\pi(\boldsymbol{x})\pi(\boldsymbol{y})}$ meets $L$.
Then the line $\boldsymbol{x}\boldsymbol{y}$ meets the face of $T$
supported by the plane $x=a$ iff the line in $xz$-plane
connecting $(x_1,x_3)$ and $(y_1,y_3)$ passes above $(a,c)$. This
requirement is expressed by the quadratic constraint
(8) in Appendix~\ref{append-geometry}.

%the (quadratic) formula:
%\[\begin{vmatrix} x_1 & y_3 & 1 \\                                         \                                                            %                                                                            
%                   a & c & 1 \\                                                                                                       
%                   y_1 & y_3 & 1                                                                                                        %  
%\end{vmatrix} > 0 \wedge x_1<a<y_1 \qquad\mbox{or}\qquad
%\begin{vmatrix} x_1 & y_3 & 1 \\                                         \                                                              %                                                                          
%                   a & c & 1 \\                                                                                                       
%                   y_1 & y_3 & 1                                                                                                        %  \end{vmatrix} < 0 \wedge y_1<a<x_1 
% \]
{\bf A bounding edge of $T$ meets $S$.}
\label{subsec:phi3}
We proceed to describe the formula $\varphi_T^{(3)}$ expressing 
that the bounding edge $E$ of $T$, supported by the line
$x=a\cap y=b$, meets $S$.  Note that image of $E$ under the projection
$\pi$ is the single point $\boldsymbol{c}=(a,b)$.  Thus $E$ meets $S$
just in case~$\boldsymbol{c} \in \pi(S)$ and the point $(a,b,c)$ lies
below the plane affinely spanned by $S$.  We describe two formulas
that respectively express these requirements.

Denote the vertices of $S$ by $\boldsymbol{p}$, $\boldsymbol{q}$, and
$\boldsymbol{r}$.  We first give a formula of linear arithmetic
expressing that $\boldsymbol{c}\in\pi(S)$.  Notice that if
$\boldsymbol{c}\in\pi(S)$ then at least one vertex of $\pi(S)$ must
lie in~$\pi(F)$.  We now consider two cases.  The first case is that
exactly one vertex of $\pi(S)$ (say~$\pi(\boldsymbol{p})$) lies in~$\pi(F)$.
 The second case is that at least two vertices of of $\pi(S)$
(say~$\pi(\boldsymbol{p})$ and $\pi(\boldsymbol{q})$) lie in $\pi(F)$.
The two cases are respectively denoted in Figure~\ref{fig:BoundingEdges3D}, that   
we refer to  in the following.

\begin{figure}[t]
 \begin{center}
   \scalebox{.85}{\begin{tikzpicture}[dot/.style={circle,inner sep=1pt,fill,label={#1},name=#1},
  extended line/.style={shorten >=-#1,shorten <=-#1},
  extended line/.default=1cm]

\node [draw=none,label=left:{{\bf Case 1:}}] (q) at (0,0) {~\scalebox{.9}{\begin{tikzpicture}
	[grid/.style={very thin,gray},
		axis/.style={->, thick},
		cube/.style={opacity=.5, thick,green,fill=green}]
	
	\coordinate (C) at (.8,.8); 
	\coordinate (Ca) at (1.2,0);
	\coordinate (Cb) at (0,1.2);
  \coordinate (orig) at (0,0);
	\coordinate (xax) at (2.4,0);
	\coordinate (yax) at (0,2.4);
	\coordinate (G1) at (intersection of C--Ca and orig--yax);
  \coordinate (G2) at (intersection of C--Cb and orig--xax);
	\coordinate (c) at (.8,.8); 
  \coordinate (p) at (.4,.55); 
	\coordinate (r) at (1.04,1.88); 
  \coordinate (q) at (2.12,1.04); 
	
	 %draw the IMAGE of F
	\draw[thick,opacity=.5,black!50!white, fill=black!40!white] (G1) -- (orig) -- (Ca) -- cycle;
	\draw[thick,opacity=.5,black!50!white, fill=black!40!white] (G2) -- (orig) -- (Cb) -- cycle;

%\draw[thick,red!50!white, fill=red!30!white] (0,0) -- (1.5,0) -- (1.5,1.5) -- (0,1.5) -- cycle;

\draw[very thin,color=gray,step=.4cm] (0,0) grid (2.8,2.8);
\draw[->,ultra thick] (-.1,0)--(2.6,0) node[right]{$x$};
\draw[->,ultra thick] (0,-.1)--(0,2.6) node[above]{$y$};
	
\draw [fill=black,black](c) circle (2pt) node [above,black]{$~~\boldsymbol{c}$};
\draw [fill=black,black](xax) circle (2pt) node [above,black]{$\boldsymbol{f_2}$};
\draw [fill=black,black](yax) circle (2pt) node [left,black]{$\boldsymbol{f_1}$};
	
\draw [fill=black,black](p) circle (2pt) node [below,black]{$\pi(\boldsymbol{p})$};
\draw [fill=black,black](q) circle (2pt) node [right,black]{$\pi(\boldsymbol{q})$};
\draw [fill=black,black](r) circle (2pt) node [right,black]{$\pi(\boldsymbol{r})$};

\draw (p) -- (q);
\draw (p) -- (r);
\draw (q) -- (r);

\end{tikzpicture}}};
\node [draw=none,label=left:{{\bf Case 2:}}] (q) at (8,0) {~\scalebox{.9}{\begin{tikzpicture}
	[grid/.style={very thin,gray},
		axis/.style={->, thick},
		cube/.style={opacity=.5, thick,green,fill=green}]
	
\coordinate (C) at (.8,.8); 
	\coordinate (Ca) at (1.2,0);
	\coordinate (Cb) at (0,1.2);
  \coordinate (orig) at (0,0);
	\coordinate (xax) at (2.4,0);
	\coordinate (yax) at (0,2.4);
	\coordinate (G1) at (intersection of C--Ca and orig--yax);
  \coordinate (G2) at (intersection of C--Cb and orig--xax);
	\coordinate (c) at (.8,.8); 
  \coordinate (p) at (.4,.55); 
	\coordinate (r) at (1.04,1.84); 
  \coordinate (q) at (1.52,.128); 
	
	 %draw the IMAGE of F
	\draw[thick,opacity=.5,black!50!white, fill=black!40!white] (G1) -- (orig) -- (Ca) -- cycle;
	\draw[thick,opacity=.5,black!50!white, fill=black!40!white] (G2) -- (orig) -- (Cb) -- cycle;

%\draw[thick,red!50!white, fill=red!30!white] (0,0) -- (1.5,0) -- (1.5,1.5) -- (0,1.5) -- cycle;

\draw[very thin,color=gray,step=.4cm] (0,0) grid (2.8,2.8);
\draw[->,ultra thick] (-.1,0)--(2.6,0) node[right]{$x$};
\draw[->,ultra thick] (0,-.1)--(0,2.6) node[above]{$y$};
	
\draw [fill=black,black](c) circle (2pt) node [above,black]{$~~\boldsymbol{c}$};
\draw [fill=black,black](xax) circle (2pt) node [above,black]{$\boldsymbol{f_2}$};
\draw [fill=black,black](yax) circle (2pt) node [left,black]{$\boldsymbol{f_1}$};
	
\draw [fill=black,black](p) circle (2pt) node [below,black]{$\pi(\boldsymbol{p})$};
\draw [fill=black,black](q) circle (2pt) node [above,black]{~~~~~~$\pi(\boldsymbol{q})$};
\draw [fill=black,black](r) circle (2pt) node [right,black]{$\pi(\boldsymbol{r})$};

\draw (p) -- (q);
\draw (p) -- (r);
\draw (q) -- (r);

\end{tikzpicture}}};

\end{tikzpicture}}
		\end{center}
		\vspace{-.4cm}
\caption{Two cases for expressing that $\boldsymbol{c}\in\pi(S)$.  The grey region is $\pi(F)$.}
\label{fig:BoundingEdges3D}
\vspace{-.4cm}
\end{figure}

In the first case we can express that $\boldsymbol{c} \in\pi(S)$ by
requiring that the line segment
$\pi(\boldsymbol{p})\pi(\boldsymbol{q})$ crosses the edge $\boldsymbol{f}_2\boldsymbol{c}$
 and $\pi(\boldsymbol{p})\pi(\boldsymbol{r})$ crosses the
edge $\boldsymbol{f}_1\boldsymbol{c}$.  By writing a separate constraint for
each fixed value of $\pi(\boldsymbol{p}) \in \pi(F)$ the above
requirements can be expressed in linear arithmetic.

In the second case we can express that $\boldsymbol{c} \in \pi(S)$ by
requiring that $\boldsymbol{c}$ lies on the left of each of the
directed line segments ${\pi(\boldsymbol{p})\pi(\boldsymbol{q})}$,
${\pi(\boldsymbol{q})\pi(\boldsymbol{r})}$, and
${\pi(\boldsymbol{r})\pi(\boldsymbol{p})}$.  By writing such a
constraint for each fixed value of $\pi(\boldsymbol{p})$ and
$\pi(\boldsymbol{q})$ in $\pi(F)$ we obtain, again, a formula of
linear arithmetic, see Appendix~\ref{append-geometry}.

It remains to give a formula expressing that~$(a,b,c)$ lies below the
plane affinely spanned by $\boldsymbol{p}$, $\boldsymbol{q}$, and
$\boldsymbol{r}$ under the assumption that
$\boldsymbol{c} \in \pi(S)$.  Note here that the above-described
formula expressing that $\pi(\boldsymbol{c})\in\pi(S)$ specifies
\emph{inter alia} that $\pi(\boldsymbol{p})$, $\pi(\boldsymbol{q})$,
and $\pi(\boldsymbol{r})$ are oriented counter-clockwise.  
Thus $(a,b,c)$ lies below the plane affinely spanned by
$\boldsymbol{p}$, $\boldsymbol{q}$, and $\boldsymbol{r}$ iff
\[ \begin{vmatrix} 
q_1-p_1&r_1-p_1&a-p_1\\
q_2-p_2&r_2-p_2&b-p_2\\
q_3-p_3&r_3-p_3&c-p_3\\ \end{vmatrix} <0 \]
%Thus we may
%use the determinant expression in Appendix~\ref{append-geometry} to
%specify that~$(a,b,c)$ lies below the plane affinely spanned by
%$\boldsymbol{p}$, $\boldsymbol{q}$, and $\boldsymbol{r}$.  This yields
The above expession is cubic, but by
Proposition~\ref{prop-3dim-fixed-a-point} we may assume that
$\boldsymbol{p}$ lies in the set $F$, which has finitely many integer
points.  Thus by a case analysis we may regard $\boldsymbol{p}$ as
being fixed and so write the desired formula as a disjunction of
atoms, each with a single quadratic term, whose satisfiability is known to
be decidable from Theorem~\ref{thm:segal}. This leads us to:
%\textcolor{red}{This then leads to Theorem~\ref{thm:3-costs} that summarizes the results of this section:}

\begin{theorem}\label{thm:3-costs} The Pareto Domination Problem is
  decidable for at most three observers.
\end{theorem}
	
Theorem~\ref{thm:3-costs} was proven
 by reduction to 
satisfiability of a system of arithmetic constraints with  a \emph{single} quadratic term. 
For the case of four observers this technique does not appear to yield arithmetic constraints in a known decidable class.
Note that satisfiability of  systems of constraints featuring two distinct quadratic terms is not known to be decidable in general.

In Appendix~\ref{subsec:two-cost-variables} we consider (a
generalisation of) the Pareto Domination Problem for MPTA with at
most two observers.  In contrast to the case of three
observers, we are able to show decidability for two observers by
reduction to satisfiability in linear arithmetic.

% \textcolor{red}{To achieve the decidability result in
%   Theorem~\ref{thm:3-costs}, we reduce the Pareto Domination Problem
%   to the existence of positive integer zeros of a quadratic form.  In
%   appendix~\ref{subsec:two-cost-variables}, we study MPTAs with two
%   observer variables and reachability of sets of valuations
%   $T\subseteq\RP^{\mathcal{Y}}$ described by arbitrary conjunctions of
%   constraints of the form $\gamma(y)\sim c$ for $y\in\mathcal{Y}$,
%   ${\sim}\in\{\leq,\geq\}$, and $c\in\mathbb{Z}$.  Since the set of
%   valuations in $\RP^{\mathcal{Y}}$ dominating a given valuation can
%   be written in the above form, this reachability problem subsumes the
%   Pareto Domination Problem.  In contrast to the situation with three
%   observer variables, in the case with two observers we translate the
%   reachability problem into satisfiability in linear arithmetic.}

\section{Gap Domination Problem} \label{sec:gap}
In this section we give a decision procedure for the Gap Domination Problem.  Given an MPTA $\mathcal{A}$, valuation~$\gamma\in\RP^{\mathcal{Y}}$, and
a rational tolerance~$\varepsilon>0$, our procedure is such that
\begin{itemize}
\item if there is an accepting run $\rho$ of $\mathcal{A}$ such that
  $\gamma\preccurlyeq_{\varepsilon} \cost(\rho)$ then we output
  ``dominated'';
\item if there is no accepting run $\rho$ of $A$ such that
$\gamma\preccurlyeq \cost(\rho)$ then we output ``not dominated''.
\end{itemize}
To do this, our approach is to find approximate solutions of the
bilinear system (\ref{eq:master-system}) by relaxation and rounding.

Recall from Proposition~\ref{prop:master} that
(\ref{eq:master-system}) is satisfiable iff $\mathcal{A}$ has an
accepting run~$\rho$ such that~$\gamma\preccurlyeq \cost(\rho)$.  Now
we use the semi-linear decomposition of
$\mathrm{Reach}_{\mathcal{S(A)}}$ to eliminate the constraints on
integer variables from (\ref{eq:master-system}). In more detail, fix
a decomposition of $\mathrm{Reach}_{\mathcal{S(A)}}$ as a union of
linear sets and let $S:=S(\boldsymbol{v},P)$ be one such linear set,
where $P=\{\boldsymbol{u}_1,\ldots,\boldsymbol{u}_k\}$.  Then we
replace the constraint
$(\gamma_1,\ldots,\gamma_{d+1}) \in  \mathrm{Reach}_{\mathcal{S(A)}}$
in (\ref{eq:master-system}) with
\[ (\gamma_1,\ldots,\gamma_{d+1}) = \boldsymbol{v} +
  n_1\boldsymbol{u}_1 + \cdots + n_k \boldsymbol{u}_k\, , \] where
$n_1,\ldots,n_k$ are variables ranging over $\mathbb{N}$.  We thus
obtain for each choice of $S$ a bilinear system of inequalities $\varphi_S$ of the form
(\ref{eq:approx-system}), where $I$ and $J$ are finite sets and for
each $i\in I$ and~$j\in J$, it holds that $f_i,g_j$ are linear forms
(i.e., polynomials of degree one with no constant terms) with
non-negative integer coefficients and $c_i$ and $d_j$ are rational
constants. 
\begin{equation}
\begin{array}{rclcrcl}
f_i(n_1\lambda_1,n_1\lambda_2,\ldots,n_k\lambda_{d+1}) & \leq & c_i \quad (i\in I)  & \quad \quad \quad& \lambda_1,\ldots,\lambda_{d+1} &\geq & 0 \\
g_j(n_1\lambda_1,n_1\lambda_2,\ldots,n_k\lambda_{d+1}) & \geq & d_j \quad (j\in J)  & \quad \quad \quad&\lambda_1+\cdots+\lambda_{d+1} &=& 1\\
{n_1,\ldots,n_k\in \mathbb{N}} & & & & &
\end{array}
\label{eq:approx-system}
\end{equation}
%Then $\A$ has an accepting run $\rho$ such that
%$\gamma\preccurlyeq \cost(\rho)$ just in case $\varphi_S$
%is satisfiable for some $S$ in the semi-linear decompositon of
%$\mathrm{Reach}(\mathcal{S}(\mathcal{A}))$.

Fix a particular system $\varphi_S$, as depicted in
(\ref{eq:approx-system}).  Let $\mu$ be the maximum coefficient of the~$f_i$,~$i\in I$. 
 Given $T\subseteq \{1,\ldots,d+1\}$, we define the
following constraint $\psi_T$ on $\lambda_1,\ldots,\lambda_{d+1}$:
\[ \psi_T := \bigwedge_{i\in T}\lambda_i \leq
  \textstyle\frac{\varepsilon}{(d+1)k\mu} \wedge \displaystyle\bigwedge_{i\not\in T} \lambda_i
  \geq \textstyle\frac{\varepsilon}{(d+1)k\mu} \, .\] Intuitively, $\psi_T$
expresses that $\lambda_i$ is ``small'' for $i\in T$ and ``large'' for
$i\not\in T$.  Given any satisfying assignment of
$\varphi_S$ it is clear that
$\lambda_1,\ldots,\lambda_{d+1}$ must satisfy 
$\varphi_T$ for some $T\subseteq \{1,\ldots,d+1\}$.  

Now fix a set $T\subseteq \{1,\ldots,d+1\}$ and consider the
satisfiability of $\varphi_S \wedge \psi_T$.  If $i\not\in T$ then for
any term $\lambda_i n_j$ that appears in an upper-bound constraint
with right-hand side $c$ in $\varphi_S$, we must have
$n_j \leq \lceil \frac{c(d+1)\mu}{\varepsilon}\rceil$ in order for the
constraint to be satisfied.  Thus by enumerating all values of $n_j$
we can eliminate this variable.  By doing this we may assume that in
$\varphi_S\wedge\psi_T$, for any term $\lambda_i n_j$ that appears on
the left-hand side of an upper-bound constraint we have~$i\in T$ and hence that $\lambda_i$ must be ``small'' in any satisfying
assignment.

The next step is relaxation---try to solve $\varphi_S \wedge \psi_T$
(after the above described elimination step), letting the variables
$n_1,\ldots,n_k$ range over the non-negative reals.  Recall here that
the existential theory of real closed fields is decidable in
polynomial space.  If there is no real solution of
$\varphi_S \wedge \psi_T$ for any $S$ and $T$ then there is certainly
no solution over the naturals. and we can output ``not dominated''.
On the other hand, if there is a run $\rho$ with
$\gamma\preccurlyeq_\varepsilon \cost(\rho)$ then for some $S$ and
$T$, the system $\varphi_S\wedge \psi_T $ will have a real solution in
which moreover the inequalities~$f_i(n_1\lambda_1,\ldots,n_k\lambda_{d+1}) \leq c_i$ for $i\in I$ all
hold with slack at least $\varepsilon$.  Given such a solution,
replace $n_j$ with $\lceil n_j \rceil$ for $j=1,\ldots,k$.  Consider
the left-hand side~$f_i(n_1\lambda_1,\ldots,n_k\lambda_{d+1})$ of an
upper bound constraint in $\varphi_S$.  Since the variables~$\lambda_i$ mentioned in such a linear form are small, the effect of
rounding is to increase this term by at most $\varepsilon$.  Hence the
rounded valuation still satisfies $\varphi_S$ thanks to the slack in
the original solution. This then leads to Theorem~\ref{thm:gap} below:

\begin{theorem}\label{thm:gap} The Gap Domination Problem is decidable.
\end{theorem}

%=====questions: find answers  or add to perspective====
%\section{Stuff to Think About}
%\input{questions}

%\section{Conclusion} \label{sec:conclusion}
%\input{Conclusion}

\bibliography{references}

\begin{thebibliography}{10}

\bibitem{AlurDill94}
R.\ Alur and D.\ Dill.
\newblock A theory of timed automata.
\newblock {\em TCS}, 126(2):183--235, 1994.

\bibitem{Alur01PTA}
R.~Alur, S.~La Torre, and G.~J. Pappas.
\newblock Optimal paths in weighted timed automata.
\newblock In M.-D.~Di Benedetto and A.~S-Vincentelli, editors, {\em HSCC},
  volume 2034 of {\em LNCS}, pages 49--62. Springer, 2001.

\bibitem{Behrmann01}
G.~Behrmann, A.~Fehnker, T.~Hune, K.~G. Larsen, P.~Pettersson, J.~Romijn, and
  F.~W. Vaandrager.
\newblock Minimum-cost reachability for priced timed automata.
\newblock In M.-D.~Di Benedetto and A.~S-Vincentelli, editors, {\em HSCC},
  volume 2034 of {\em LNCS}, pages 147--161. Springer, 2001.

\bibitem{Bouyer07}
P.~Bouyer, T.~Brihaye, V.~Bruy{\`e}re, and J.-F. Raskin.
\newblock On the optimal reachability problem of weighted timed automata.
\newblock {\em Formal Methods in System Design}, 31(2):135--175, 2007.

\bibitem{Bouyer08b}
P.~Bouyer, E.~Brinksma, and K.~G. Larsen.
\newblock Optimal infinite scheduling for multi-priced timed automata.
\newblock {\em Formal Methods in System Design}, 32(1):3--23, 2008.

\bibitem{Bouyer08d}
P.~Bouyer, U.~Fahrenberg, K.~G. Larsen, N.~Markey, and J.~Srba.
\newblock Infinite runs in weighted timed automata with energy constraints.
\newblock In F.~Cassez and C.~Jard, editors, {\em FORMATS}, volume 5215 of {\em
  LNCS}, pages 33--47. Springer, 2008.

\bibitem{Bouyer08c}
P.~Bouyer, K.~G. Larsen, and N.~Markey.
\newblock Model checking one-clock priced timed automata.
\newblock {\em Logical Methods in Computer Science}, 4:1--28, 2008.

\bibitem{BrihayeBruyereRaskin06}
T.~Brihaye, V.~Bruy{\`{e}}re, and J.{-}F. Raskin.
\newblock On model-checking timed automata with stopwatch observers.
\newblock {\em Inf. Comput.}, 204(3):408--433, 2006.

\bibitem{DiakonikolasY09}
I.~Diakonikolas and M.~Yannakakis.
\newblock Small approximate pareto sets for biobjective shortest paths and
  other problems.
\newblock {\em {SIAM} J. Comput.}, 39(4):1340--1371, 2009.

\bibitem{FORMATS09}
M.~Fr{\"{a}}nzle and M.~Swaminathan.
\newblock Revisiting decidability and optimum reachability for multi-priced
  timed automata.
\newblock In J.~Ouaknine and F.~W. Vaandrager, editors, {\em {FORMATS}}, volume
  5813 of {\em LNCS}, pages 149--163. Springer, 2009.

\bibitem{GrunewaldSegal04}
F.~Grunewald and D.~Segal.
\newblock On the integer solutions of quadratic equations.
\newblock {\em Journal f\"ur die reine und angewandte Mathematik}, 569:13--45,
  2004.

\bibitem{HaaseHalfon14}
C.~Haase and S.~Halfon.
\newblock Integer vector addition systems with states.
\newblock In J.~Ouaknine, I.~Potapov, and J.~Worrell, editors, {\em RP}, volume
  8762 of {\em LNCS}, pages 112--124. Springer, 2014.

\bibitem{Henzinger98}
T.~A. Henzinger, P.~W. Kopke, A.~Puri, and P.~Varaiya.
\newblock What's decidable about hybrid automata?
\newblock {\em J. Comput. Syst. Sci.}, 57(1):94--124, 1998.

\bibitem{Jones80}
J.~P. Jones.
\newblock Undecidable diophantine equations.
\newblock {\em Bull. Amer. Math. Soc.}, 3:859--862, 1980.

\bibitem{KopTo10}
E.~Kopczynski and A.~W. To.
\newblock Parikh images of grammars: Complexity and applications.
\newblock In {\em LICS}, pages 80--89. {IEEE} Computer Society, 2010.

\bibitem{Larsen01}
K.~G. Larsen, G.~Behrmann, E.~Brinksma, A.~Fehnker, T.~Hune, P.~Pettersson, and
  J.~Romijn.
\newblock As cheap as possible: Efficient cost-optimal reachability for priced
  timed automata.
\newblock In G.~Berry, H.~Comon, and A.~Finkel, editors, {\em CAV}, volume 2102
  of {\em LNCS}, pages 493--505. Springer, 2001.

\bibitem{Larsen08}
K.~G. Larsen and J.~I. Rasmussen.
\newblock Optimal reachability for multi-priced timed automata.
\newblock {\em TCS}, 390(2-3):197--213, 2008.

\bibitem{Perevoshchikov15}
V.~Perevoshchikov.
\newblock {\em Multi-weighted automata models and quantitative logics.}
\newblock PhD thesis, University of Leipzig, 2015.

\bibitem{Quaas10}
K.~Quaas.
\newblock {\em Kleene-Sch{\"{u}}tzenberger and B{\"{u}}chi theorems for
  weighted timed automata}.
\newblock PhD thesis, University of Leipzig, 2010.

\bibitem{Lin10}
A.~W. To.
\newblock Parikh images of regular languages: Complexity and applications.
\newblock {\em CoRR}, 2010.
\newblock URL: \url{http://arxiv.org/abs/1002.1464}.

\bibitem{Zaslavsky82}
T.~Zaslavsky.
\newblock Signed graphs.
\newblock {\em Discrete Applied Mathematics}, 4(1):47 -- 74, 1982.

\bibitem{Zhang17}
Z.~Zhang, B.~Nielsen, K.~G. Larsen, G.~Nies, M.~Stenger, and H.~Hermanns.
\newblock Pareto optimal reachability analysis for simple priced timed
  automata.
\newblock In Z.~Duan and L.~Ong, editors, {\em {ICFEM}}, volume 10610 of {\em
  LNCS}, pages 481--495. Springer, 2017.

\end{thebibliography}

%===== Appendix
\newpage
\appendix
\section{Difference Constraints}\label{app:difference}

As summarized in~\cite[Section 5.3]{Bouyer07} for the setting of a
single observer, given an MPTA~$\mathcal{A}$ with difference clock
constraints, we can find an MPTA~$\mathcal{A'}$ without difference
clock constraints such that $\mathcal{A}$ and $\mathcal{A'}$ are
strongly time-bisimilar.  The Domination Problems  for
$\mathcal{A}$ can thus be reduced to those for $\mathcal{A'}$.
Although eliminating difference clock constraints from MPTA results in
an exponential blow-up in the number of locations and edges~\cite[Section 5.3]{Bouyer07},
the PSPACE complexity
for the Pareto Domination Problem in the case of all cost variables
and all reward variables (see Section~\ref{sec:pure-const} and
Appendix~\ref{append-all-reward}) remains true.  Indeed the granularity
bounds that were used to establish PSPACE complexity, while
exponential in the number of observers, are only polynomial
in the number of locations of the MPTA and hence remain singly
exponential in magnitude even after an exponential blow-up in the
number of locations.

\section{Missing Proofs}

\subsection{Proof of Proposition~\ref{prop:zvass}} \label{append-zvass}
\begin{proof}
  Given $\mathcal{Z}$ and $q$, we can construct an NFA $\mathcal{B}$ over alphabet
  $\Sigma' = \{\sigma_1,\ldots,\sigma_n\}$ with at most
  $|Q|^2nM$ states such that $\mathrm{Reach}_{\mathcal{Z}}$ is the Parikh image of
  the language of $\mathcal{B}$. The idea is that each transition
  $(p,\boldsymbol{v},p')$ in $\mathcal{A}$ is simulated in $\mathcal{B}$ by a gadget
  consisting of a sequence of transitions whose Parikh image is
  $\boldsymbol{v}$.

  Having obtained $\mathcal{B}$, the proposition follows from the bound
 in~\cite[Proposition 4.3]{Lin10},\cite{KopTo10} on the size of the semilinear
  decomposition of the Parikh image of the language of an NFA.
\end{proof}

\subsection{Proof that Satisfiability for Language $\mathcal{L}$ is
  Undecidable (Section~\ref{sec:undecidability})}
\label{append-variant-Diophantine-prop}
\begin{proposition} \label{variant-Diophantine-prop}                                                                                     
The satisfiability problem for $\mathcal{L}$ is undecidable.                                                                             
\end{proposition} 

\begin{proof}
The proof is by reduction from Hilbert's Tenth Problem: given a
  polynomial $P~\in \mathbb{Z}[X_1,\ldots,X_k]$, does $P$ have a zero
  over the set of positive integers?  Given such a polynomial~$P$, we
  write an $\mathcal{L}$-formula~$\varphi_P$ whose variables include
  $X_1,\ldots,X_k$, such that the satisfying assignments
  of~$\varphi_P$ are in one-to-one correspondence with the positive
  integer roots of $P$.  

  The idea is simple: write $P=P_1-P_2$, where all monomials in $P_1$
  and $P_2$ appear with positive coefficients.  We then introduce an
  $\mathcal{L}$-variable for each subterm of $P_1$ and $P_2$ and write
  constraints to ensure that the variable takes the same value as the
  corresponding term. Finally we assert that $P_1$ is equal to $P_2$
  through the constraint $P_1 = P_2 X \wedge X = XX$.
\end{proof}

\subsection{Proof of Proposition~\ref{prop:diff}}\label{append-prop-diff}
\begin{proof}
Given a sequence $(t_1,\ldots,t_m) \in \RP^m$, we define a
  corresponding sequence of clock valuations
  $\nu_1,\ldots,\nu_m \in \RP^{\mathcal{X}}$ by $\nu_i(x)=t_i$ if
  none of the edges $e_1,\ldots,e_{i-1}$ reset clock $x$ and
  otherwise $\nu_i(x) := t_i-t_j$, where $j<i$ is the maximum index
  such that $x$ is reset by edge~$e_j$.  In order for a sequence
  $(t_1,\ldots,t_m)$ to be an element of
  $\mathit{Runs}(e_1,\ldots,e_m)$ we require that the~$t_i$ be
  non-negative and non-decreasing and that for every index
  $i\in\{1,\ldots,m\}$, the guard $\varphi_i$ of edge $e_i$ be
  satisfied by the clock valaution $\nu_i$ defined above.  Clearly the
  above requirements can be expressed by difference constraints on
  $t_1,\ldots,t_m$.
\end{proof}

\subsection{Proof of Proposition~\ref{prop:SL-decomp}}\label{append-prop-SL-decomp}
\begin{proof}
  The number of control states of $\mathcal{Z(A)}$ is at most
  $(M_{\mathcal{X}})^{|\mathcal{X}|}|L|$ and the number of states of
  $\mathcal{S(A)}$ is at most
  $((M_{\mathcal{X}})^{|\mathcal{X}|}|L|)^{d+1}$.  Moreover the
  vectors occurring in the transitions of $\mathcal{S(A)}$ have
  entries of magnitude at most $M_{\mathcal{Y}}M_{\mathcal{X}}$.  We
  now apply Proposition~\ref{prop:LIN} to $\mathcal{S(A)}$.  We get
  that the the base vectors~$\boldsymbol{v}_i$ and period vectors in
  $P_i$ have entries of magnitude at most
  $\mathit{poly}(d,|L|,M_{\mathcal{Y}},M_{\mathcal{X}})^{d(d+1)|\mathcal{X}|}$.
\end{proof}

\subsection{Proof of Proposition~\ref{prop:case-split}}\label{append-prop-case-split}
\begin{proof}
Observe that $T\cap S$ is nonempty just in case there exists a point
$\boldsymbol{x}=(x_1,x_2,x_3) \in S$ such that
$\pi(\boldsymbol{x})\in\pi(T)\cap\pi(S)$ and $x_3\geq c$.  But
$\pi(T) \cap \pi(S)$, being a bounded convex polygon, is the convex
hull of its vertices.  It follows that $T\cap S$ is non-empty just in
case there exists a point $\boldsymbol{x} \in S$ such that
$\pi(\boldsymbol{x})$ is a \emph{vertex} of $\pi(T)\cap\pi(S)$ and
$x_3\geq c$.

Now the vertices of $\pi(T)\cap \pi(S)$ come in three types:
$(i)$~vertices of~$\pi(S)$, $(ii)$~intersections of bounding
line segments of $\pi(T)$ and $\pi(S)$, and $(iii)$~vertices
of~$\pi(T)$.

Let $\boldsymbol{x} \in S$ be such that $\pi(\boldsymbol{x})$ is a
vertex of $\pi(T)\cap\pi(S)$ and $x_3\geq c$.  Assume moreover that
for all $\boldsymbol{y} \in S$ such that
$\pi(\boldsymbol{x})=\pi(\boldsymbol{y})$ we have $x_3 \geq y_3$.  If
$\pi(\boldsymbol{x})$ is a vertex of $\pi(T)\cap\pi(S)$ of the first
type then $\boldsymbol{x}$ is a vertex of $S$.  If
$\pi(\boldsymbol{x})$ is a vertex of the second type, but not of the
first type, then $\boldsymbol{x}$ is the intersection of a bounding
edge of $S$ with one of the two faces of $F$ identified in Item 2 in
the statement of the proposition. Finally, if $\pi(\boldsymbol{x})$ is
a vertex of the third type, but not of the first or second types, then
$\boldsymbol{x}$ is the intersection of $S$ with the edge of $F$
supported by the line $x=a\cap y=b$.
\end{proof}

\subsection{Proof of Proposition~\ref{prop-3dim-fixed-a-point}}\label{append-fixed-a-point}
\begin{proof}
  Since $S\cap T\neq\emptyset$, we have
  $\pi(S)\cap\pi(T)\neq\emptyset$.  Hence there are vertices
  $\boldsymbol{x},\boldsymbol{y}$ of $S$ such that the edge
  $\pi(\boldsymbol{x})\pi(\boldsymbol{y})$ meets $\pi(T)$.  By
  Proposition~\ref{prop:decomposition-of-F} we have either
  that one of~$\pi(\boldsymbol{x})$ and~$\pi(\boldsymbol{y})$ lies in
  $\pi(T)$ or that both $\pi(\boldsymbol{x})$ and
        $\pi(\boldsymbol{y})$ lie in $\pi(F)$.

  Suppose $\pi(\boldsymbol{x}) \in \pi(T)$.  Since the edge
  $\boldsymbol{x}\boldsymbol{y}$ is assumed not to meet $T$ we must
  have that~$x_3<c$ and hence $\boldsymbol{x}\in F$.  Likewise the
  assumption that $\pi(\boldsymbol{y}) \in \pi(T)$ yields
  $\boldsymbol{y}\in F$.  Finally, if both $\pi(\boldsymbol{x})$ and
  $\pi(\boldsymbol{y})$ lie in $\pi(F)$ then the assumption that
  $\boldsymbol{x}\boldsymbol{y}$ does not meet $T$ implies that either
  $x_3< c$ or $y_3<c$.  Hence $\boldsymbol{x}\in F$ or
  $\boldsymbol{y}\in F$.
\end{proof}

\section{Pareto Domination with All Reward Variables}\label{append-all-reward}
Now we suppose that the set of observers $\mathcal{Y}$ is comprised
exclusively of reward variables.  We will again apply
Proposition~\ref{prop:SL-decomp} to rewrite (\ref{eq:master-system}) as a
finite disjunction of systems of linear inequalities.

Fix an index $i \in I$.  Let the base vector of the linear set
$S(\boldsymbol{v}_i,P_i)$ be
$\boldsymbol{v}_i=(\gamma_1,\ldots,\gamma_{d+1})$.  We write a linear
constraint to express that there exists a vector
$(\gamma'_1,\ldots,\gamma'_{d+1}) \in S(\boldsymbol{v}_i,P_i)$ and a
convex combination $\sum_{j=1}^{d+1} \lambda_j \gamma'_j$ that
dominates a given $\gamma\in\RP^{\mathcal{Y}}$.  We write this
constraint as a disjunction of finitely many systems of linear
inequalities---one system for each possible choice of the support
$S'\subseteq\{1,\ldots,d+1\}$ of the the convex sum.  Fix such a set
$S'$ and let~$\mathcal{Y}_{S'} \subseteq \mathcal{Y}$ be the set of
variables $y$ such that there is some period vector
$(\gamma'_1,\ldots,\gamma'_{d+1}) \in P_i$ and $j\in S'$ with
$\gamma'_j(y)>0$.  Then the system of inequalities is as follows:
\begin{equation}
\begin{array}{rcl}
\gamma(y) & \preccurlyeq & \lambda_1\gamma_1(y) + \ldots + \lambda_{d+1}\gamma_{d+1}(y) \quad (y\not\in \mathcal{Y}_{S'})\\
  1 & =& \lambda_1 + \cdots + \lambda_{d+1} \\
  0 & < & \lambda_j \quad (j \in S')\\
  0 & = & \lambda_j \quad (j\not\in S')
\end{array}
\label{eq:LP-2}
\end{equation}
To see why this works, note that for $y\in \mathcal{Y}_{S'}$ there exists
some period vector $(\gamma'_1,\ldots,\gamma'_{d+1}) \in P_i$ and
$j\in S'$ with $\gamma'_j(y)>0$.  By adding suitable multiples of to
the solution of the above system we can make value of the variable $y$
arbitrarily large.

Recall that if a set of linear inequalities
$A\boldsymbol{x}\geq \boldsymbol{a}$, $B\boldsymbol{x}>\boldsymbol{b}$
is feasible then it is satisfied by some
$\boldsymbol{x}\in\mathbb{Q}^n$ of bit-length $\mathit{poly}(n,b)$,
where $b$ is the total bit-length of the entries of~$A$,~$B$,~$\boldsymbol{a}$, and $\boldsymbol{b}$.  Applying this bound and
Proposition~\ref{prop:SL-decomp} we see that a solution of
(\ref{eq:LP-2}) can be written in the form
$\lambda_1=\frac{p_1}{g},\ldots,\lambda_{d+1}=\frac{p_{d+1}}{g}$ for
integers $p_1,\ldots,p_{d+1},g$ of bit-length at most
$\mathit{poly}(d,|L|,\log(M_{\mathcal{Y}}),\log(M_{\mathcal{X}}))$.  This entails that the
cost vector $\lambda_1\gamma_1 + \ldots + \lambda_{d+1}\gamma_{d+1}$
arises from a run of $\mathcal{A}$ with granularity $\frac{1}{g}$.

\section{Geometry Background}\label{append-geometry}
We will need the following elementary geometric facts.

Let $\boldsymbol{v}_i=(x_i,y_i)$ with $i\in\{1,2,3,4\}$
be four distinct points in~$\R^2$.
Consider the determinant 
\[ \Delta(\boldsymbol{v}_1,\boldsymbol{v}_2,\boldsymbol{v}_3) = \begin{vmatrix} x_1 & y_1 & 1 \\                                                                                         
                   x_2 & y_2 & 1 \\                                                                                                  
                   x_3 & y_3 & 1                                                                                                     
\end{vmatrix}  \]
involving three points $\boldsymbol{v}_1,\boldsymbol{v}_2$ and $\boldsymbol{v}_3$. 
Then $\Delta(\boldsymbol{v}_1,\boldsymbol{v}_2,\boldsymbol{v}_3)=0$ if and only if
the three points~$\boldsymbol{v}_1,\boldsymbol{v}_2$ and $\boldsymbol{v}_3$ are colinear,
and $\Delta(\boldsymbol{v}_1,\boldsymbol{v}_2,\boldsymbol{v}_3)>0$ if and only if
$\boldsymbol{v}_3$ lies on the right of the directed line passing through 
$\boldsymbol{v}_1$ and $\boldsymbol{v}_2$.

We say that two line segments \emph{properly intersect} if they meet at a single point
that is not an end point of either line segment.
The line segment~${\boldsymbol{v}_1\boldsymbol{v}_2}$ properly intersects
the line segment~${\boldsymbol{v}_3\boldsymbol{v}_4}$ 
if and only if the following two conditions hold:
\begin{enumerate}
	\item $\boldsymbol{v}_3$ and $\boldsymbol{v}_4$ are on the opposite sides of
the line passing through~$\boldsymbol{v}_1$ and $\boldsymbol{v}_2$:
	\[(\Delta(\boldsymbol{v}_1,\boldsymbol{v}_2,\boldsymbol{v}_3)>0 \wedge 
	\Delta(\boldsymbol{v}_1,\boldsymbol{v}_2,\boldsymbol{v}_4)<0) \vee 
	(\Delta(\boldsymbol{v}_1,\boldsymbol{v}_2,\boldsymbol{v}_3)<0 \wedge 
	\Delta(\boldsymbol{v}_1,\boldsymbol{v}_2,\boldsymbol{v}_4)>0),\]
	\item $\boldsymbol{v}_1$ and $\boldsymbol{v}_2$ are on the opposite sides of
the line passing through~$\boldsymbol{v}_3$ and $\boldsymbol{v}_4$:
	\[(\Delta(\boldsymbol{v}_3,\boldsymbol{v}_4,\boldsymbol{v}_1)>0 
	\wedge \Delta(\boldsymbol{v}_3,\boldsymbol{v}_4,\boldsymbol{v}_2)<0)
	\vee (\Delta(\boldsymbol{v}_3,\boldsymbol{v}_4,\boldsymbol{v}_1)<0 
	\wedge \Delta(\boldsymbol{v}_3,\boldsymbol{v}_4,\boldsymbol{v}_2)>0).\]
\end{enumerate}
For use in Section~\ref{sec:three-cost-variables} and Appendices \ref{append-three-var} and 
\ref{subsec:two-cost-variables} we note that if $\boldsymbol{v}_1,\boldsymbol{v}_2$ and
$\boldsymbol{v}_3$ are fixed, then the constraint 
expressing that~${\boldsymbol{v}_1\boldsymbol{v}_2}$
and~${\boldsymbol{v}_3\boldsymbol{v}_4}$ properly meet is a formula of
linear arithmetic in variables~$x_4$ and $y_4$.

Let us also note that line segment $\boldsymbol{v}_1,\boldsymbol{v}_2$
properly intersects the half-line parallel to the~$x$-axis with lower
endpoint having coordinates $(a,c)$ if and only if the following
constraint holds:
\begin{gather} \left(\begin{vmatrix} x_1&y_1&1 \\ 
                   a&c&1 \\
                   x_2&y_2&1 \end{vmatrix}>0 \mbox{ and } x_1<x_3<x_2\right)
\quad\mbox{or}\quad
\left(\begin{vmatrix} x_1&y_1&1 \\ 
                   a&c&1 \\
                   x_2&y_2&1 \end{vmatrix}<0 \mbox{ and } x_2<x_3<x_1\right)
\label{eq:could-be-linear}
\end{gather}

Let $\boldsymbol{v}_i=(x_i,y_i,z_i)$ with $i\in\{1,2,3,4\}$
be four distinct points in~$\R^3$. 
Assume that the list of vertices $\boldsymbol{v}_1,\boldsymbol{v}_2,\boldsymbol{v}_3$
describes a triangle with anti-clockwise orientation.
 Consider the determinant 
\[ \Delta = \begin{vmatrix} 
							x_2-x_1 & x_3-x_1 & x_4-x_1 \\ 
							y_2-y_1 & y_3-y_1 & y_4-y_1 \\               
							z_2-z_1 & z_3-z_1 & z_4-z_1 \\                                                                                    
\end{vmatrix} \, . \] 
Then $\Delta=0$ if and only if the point~$\boldsymbol{v}_4$ lies in the
plane affinely spanned by the three
points~$\boldsymbol{v}_1,\boldsymbol{v}_2$ and $\boldsymbol{v}_3$, and
$\Delta>0$ if and only if~$\boldsymbol{v}_4$ lies above that plane.
For use in Section~\ref{sec:three-cost-variables} and Appendix~\ref{append-three-var} we note that if
$\boldsymbol{v}_1$ and $\boldsymbol{v}_4$ are fixed, then the
constraint expressing that $\boldsymbol{v}_4$ lies above the plane affinely spanned
by~$\boldsymbol{v}_1,\boldsymbol{v}_2$ and
$\boldsymbol{v}_3$ is a quadratic formula in the 
variables~$x_2,y_2,x_3$ and $y_3$.

\section{Pareto Domination with Three Mixed Observers: Two Reward
  Variables and One Cost Variable} \label{append-three-var}

Recall the set $F$, defined in Equation (\ref{eq:def-set-F}) and
consider its projection $\pi(F)$ in the $xy$-plane.  Moreover write
$R:= \{(x,y) \in \RP^2 : x\leq a \wedge y\leq b \}$ (see
Figure~\ref{fig:regionF}).

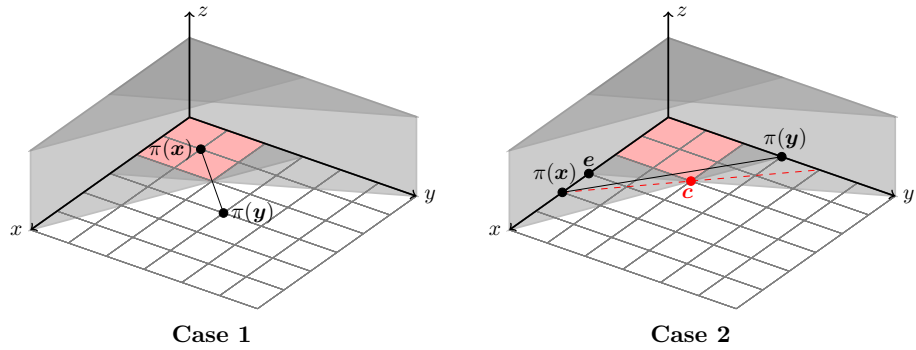
\begin{figure}[H]
 \begin{center}
   \scalebox{.9}{\begin{tikzpicture}[dot/.style={circle,inner sep=1pt,fill,label={#1},name=#1},
  extended line/.style={shorten >=-#1,shorten <=-#1},
  extended line/.default=1cm]

\node [draw=none,label=below:{{\bf Case 1}}] (q) at (0,0) {~\scalebox{.9}{\tdplotsetmaincoords{60}{125}
\begin{tikzpicture}
	[tdplot_main_coords,
		grid/.style={very thin,gray},
		axis/.style={->, thick},
		cube/.style={opacity=.5, thick,green,fill=green}]
	
	\coordinate (C) at (1.5,1.5,0); 
	\coordinate (Ca) at (2.25,0,0);
	\coordinate (Cb) at (0,2.25,0);
  \coordinate (orig) at (0,0,0);
	\coordinate (xax) at (4,0,0);
	\coordinate (yax) at (0,4,0);
	\coordinate (G1) at (intersection of C--Ca and orig--yax);
  \coordinate (G2) at (intersection of C--Cb and orig--xax);
	\coordinate (X0) at (.75,.75,0); 
  \coordinate (Y0) at (2.25,2.25,0); 
	
	 %draw the IMAGE of F
	\draw[thick,opacity=.5,black!50!white, fill=black!40!white] (G1) -- (orig) -- (Ca) -- cycle;
	\draw[thick,opacity=.5,black!50!white, fill=black!40!white] (G2) -- (orig) -- (Cb) -- cycle;
	\draw[thick,opacity=.5,black!50!white, fill=black!40!white] (4.5,0,1.5) -- (0,0,1.5) -- (0,2.25,1.5) -- cycle;
	\draw[thick,opacity=.5,black!50!white, fill=black!40!white] (0,4.5,1.5) -- (0,0,1.5) -- (2.25,0,1.5) -- cycle;
	\draw[thick,opacity=.5,black!50!white, fill=black!40!white] (4.5,0,0) -- (4.5,0,1.5) -- (0,0,1.5) -- (orig)--cycle;
  \draw[thick,opacity=.5,black!50!white, fill=black!40!white] (0,4.5,0) -- (0,4.5,1.5) -- (0,0,1.5) -- (orig)--cycle;
	
		%%draw the bottom of the cube
	%\draw[cube] (0,0,1.5) -- (1.5,0,1.5) -- (1.5,1.5,1.5) -- (0,1.5,1.5) -- cycle;
	%
	%%draw the back-right of the cube
	%\draw[cube] (0,0,1.5) -- (0,0,3) -- (0,1.5,3) -- (0,1.5,1.5) -- cycle;
%
	%%draw the back-left of the cube
	%\draw[cube] (0,0,1.5) -- (0,0,3) -- (1.5,0,3) -- (1.5,0,1.5) -- cycle;
%
  %\draw[cube]  (1.5,1.5,1.5) --(1.5,1.5,3);
	%\draw[green!90!black, thick, dashed]  (0,1.5,3) --(1.5,1.5,3);
  %\draw[green!90!black, thick, dashed]  (1.5,0,3) --(1.5,1.5,3);
	
\draw[thick,red!50!white, fill=red!30!white] (0,0,0) -- (1.5,0,0) -- (1.5,1.5,0) -- (0,1.5,0) -- cycle;

%draw a grid in the x-y plane
	\foreach \x in {0,.75,...,4.5}
		\foreach \y in {0,.75,...,4.5}
		{
			\draw[grid] (\x,0) -- (\x,4.5);
			\draw[grid] (0,\y) -- (4.5,\y);
		}			

	%draw the axes
	\draw[axis] (0,0,0) -- (4.5,0,0) node[left]{$x$};
	\draw[axis] (0,0,0) -- (0,4.5,0) node[anchor=west]{$y$};
	\draw[axis] (0,0,0) -- (0,0,2) node[anchor=west]{$z$};
	
	\draw [fill=black,black](X0) circle (2pt) node [left,black]{$\pi(\boldsymbol{x})$};
\draw [fill=black,black](Y0) circle (2pt) node [right,black]{$\pi(\boldsymbol{y})$};
\draw (X0) -- (Y0);

\end{tikzpicture}}};
\node [draw=none,label=below:{{\bf Case 2}}] (q) at (7,0) {~\scalebox{.9}{\tdplotsetmaincoords{60}{125}
\begin{tikzpicture}
	[tdplot_main_coords,
		grid/.style={very thin,gray},
		axis/.style={->, thick},
		cube/.style={opacity=.5, thick,green,fill=green}]
	
	\coordinate (C) at (1.5,1.5,0); 
	\coordinate (Ca) at (2.25,0,0);
	\coordinate (Cb) at (0,2.25,0);
  \coordinate (orig) at (0,0,0);
	\coordinate (xax) at (4,0,0);
	\coordinate (yax) at (0,4,0);
	\coordinate (G1) at (intersection of C--Ca and orig--yax);
  \coordinate (G2) at (intersection of C--Cb and orig--xax);
	\coordinate (X0) at (3,0,0); 
  \coordinate (Y0) at (0,2.25,0); 
	
	 %draw the IMAGE of F
	\draw[thick,opacity=.5,black!50!white, fill=black!40!white] (G1) -- (orig) -- (Ca) -- cycle;
	\draw[thick,opacity=.5,black!50!white, fill=black!40!white] (G2) -- (orig) -- (Cb) -- cycle;
	\draw[thick,opacity=.5,black!50!white, fill=black!40!white] (4.5,0,1.5) -- (0,0,1.5) -- (0,2.25,1.5) -- cycle;
	\draw[thick,opacity=.5,black!50!white, fill=black!40!white] (0,4.5,1.5) -- (0,0,1.5) -- (2.25,0,1.5) -- cycle;
	\draw[thick,opacity=.5,black!50!white, fill=black!40!white] (4.5,0,0) -- (4.5,0,1.5) -- (0,0,1.5) -- (orig)--cycle;
  \draw[thick,opacity=.5,black!50!white, fill=black!40!white] (0,4.5,0) -- (0,4.5,1.5) -- (0,0,1.5) -- (orig)--cycle;
	
		%%draw the bottom of the cube
	%\draw[cube] (0,0,1.5) -- (1.5,0,1.5) -- (1.5,1.5,1.5) -- (0,1.5,1.5) -- cycle;
	%
	%%draw the back-right of the cube
	%\draw[cube] (0,0,1.5) -- (0,0,3) -- (0,1.5,3) -- (0,1.5,1.5) -- cycle;
%
	%%draw the back-left of the cube
	%\draw[cube] (0,0,1.5) -- (0,0,3) -- (1.5,0,3) -- (1.5,0,1.5) -- cycle;
%
  %\draw[cube]  (1.5,1.5,1.5) --(1.5,1.5,3);
	%\draw[green!90!black, thick, dashed]  (0,1.5,3) --(1.5,1.5,3);
  %\draw[green!90!black, thick, dashed]  (1.5,0,3) --(1.5,1.5,3);
	
\draw[thick,red!50!white, fill=red!30!white] (0,0,0) -- (1.5,0,0) -- (1.5,1.5,0) -- (0,1.5,0) -- cycle;

%draw a grid in the x-y plane
	\foreach \x in {0,.75,...,4.5}
		\foreach \y in {0,.75,...,4.5}
		{
			\draw[grid] (\x,0) -- (\x,4.5);
			\draw[grid] (0,\y) -- (4.5,\y);
		}			

	%draw the axes
	\draw[axis] (0,0,0) -- (4.5,0,0) node[left]{$x$};
	\draw[axis] (0,0,0) -- (0,4.5,0) node[anchor=west]{$y$};
	\draw[axis] (0,0,0) -- (0,0,2) node[anchor=west]{$z$};
	\draw[red, dashed] (X0) -- (0,3,0);
	\draw [fill=black,black](X0) circle (2pt) node [above,black]{$\pi(\boldsymbol{x})~~$};
	\draw [fill=black,black](Ca) circle (2pt) node [above,black]{$\boldsymbol{e}$};
		\draw [fill=red,red](C) circle (2pt) node [below,red]{$\boldsymbol{c}~$};
\draw [fill=black,black](Y0) circle (2pt) node [above,black]{$~\pi(\boldsymbol{y})$};
\draw (X0) -- (Y0);

\end{tikzpicture}}};

\end{tikzpicture}}
		\end{center}
\caption{Two cases in the proof of Proposition~\ref{prop:decomposition-of-F},
where the grey region is~$F$ and the pink region is~$R$.}
\label{fig:regionF}
\end{figure}

\begin{proposition}
  Let $L$ be an edge in $\RP^2$ that
  intersects $R$.  Then $L$ has either one endpoint in $R$ or has both
  endpoints in $\pi(F)$.
\label{prop:decomposition-of-F}
\end{proposition}
\begin{proof}
  Let $L$ have endpoints $\boldsymbol{x},\boldsymbol{y} \in \RP^2$.
  Since the complement of $\pi(F)$ is a convex region in $\RP^2$ that
  excludes $R$, at least one of $\boldsymbol{x}$ or
  $\boldsymbol{y}$ lies in $\pi(F)$.  Without loss of generality,
  assume that $\boldsymbol{x} \in \pi(F)$.  To prove the
  proposition it suffices to show that if $\boldsymbol{x}\not\in R$
  then both~$\boldsymbol{x},\boldsymbol{y} \in \pi(F)$.

  Suppose $\boldsymbol{x} \not\in R$.  Now
  $\pi(F) \setminus R=F_0 \cup F_1$, where
  $F_0=\{(x,y) \in \RP^2 \mid y+bx\leq b(a+1) \text{ and } x\geq a\}$
  and
  $F_1=\{(x,y) \in \RP^2 \mid x+ay\leq a(b+1) \text{ and } y\geq b\}$.
  Thus $\boldsymbol{x}$ lies in either $F_0$ or $F_1$.  We show
  that~$\boldsymbol{x} \in F_i$ only if
  $\boldsymbol{y} \in F_{1-i}$ for $i\in \{0,1\}$ and conclude
  that both~$\boldsymbol{x},\boldsymbol{y}\in F$.

Assume that $\boldsymbol{x}\in F_0$. Since the edge
$\boldsymbol{x}\boldsymbol{y}$ meets $R$, clearly
$\boldsymbol{y} \not\in F_0$.  Draw a line through~$\boldsymbol{x}$
and $\boldsymbol{c}$, shown as the dashed red line in the diagram.
The point $\boldsymbol{y}$ is below this line for
otherwise edge~$\boldsymbol{x}\boldsymbol{y}$ fails to
meet~$R$.  Consider the point~$\boldsymbol{e}=(0,b+1)$.  Then
the edges ${\boldsymbol{e}\boldsymbol{c}}$ and~$\boldsymbol{x}\boldsymbol{c}$ meet at~$\boldsymbol{c}$.
Since edge $\boldsymbol{x}\boldsymbol{c}$ intersects the
$x$-axis above~$\boldsymbol{e}$, it intersects the $y$-axis below the
edge~${\boldsymbol{e}\boldsymbol{c}}$, i.e. in~$\pi(F)$.  We
conclude that $\boldsymbol{y}\in F_1$.

The argument for the case $\boldsymbol{x} \in F_1$ is symmetric.
Thus we have shown that $\boldsymbol{x}),\boldsymbol{y} \in \pi(F)$.
\end{proof}

%\subsection{Pareto Domination with Two Rewards and One Cost Variable}

Consider a reachability objective $T \subseteq \RP^3$ given by two
upper-bound constraints and one lower-bound constraint, 
see Figure~\ref{fig:three-Dim-Targ-Second}.
  Write
\[ T = \{ (x,y,z) \in \RP^3 : x \geq a \wedge y \geq b \wedge z \leq c
  \} \, , \] where $a,b,c$ are non-negative integer constants.  We
write a quantifier-free first-order formula~$\varphi_T$ of arithmetic
expressing that a 3-simplex $S\subseteq \RP^3$ meets $T$.  This
formula has nine free variables: one for each of the coordinates of
the three vertices of $S$.

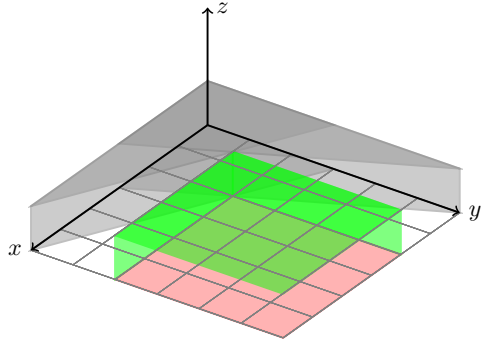
\begin{figure}[H]
 \begin{center}
   \scalebox{.9}{\tdplotsetmaincoords{60}{125}
\begin{tikzpicture}
	[tdplot_main_coords,
		grid/.style={very thin,gray},
		axis/.style={->, thick},
		cube/.style={opacity=.5, thick,green,fill=green}]
	
		\coordinate (C) at (1.5,1.5,0); 
	\coordinate (Ca) at (2.25,0,0);
	\coordinate (Cb) at (0,2.25,0);
  \coordinate (orig) at (0,0,0);
	\coordinate (xax) at (4,0,0);
	\coordinate (yax) at (0,4,0);
	\coordinate (G1) at (intersection of C--Ca and orig--yax);
  \coordinate (G2) at (intersection of C--Cb and orig--xax);
	\coordinate (X0) at (.75,.75,0); 
  \coordinate (Y0) at (2.25,2.25,0);

	 %draw the IMAGE of F
	\draw[thick,opacity=.5,black!50!white, fill=black!40!white] (G1) -- (orig) -- (Ca) -- cycle;
	\draw[thick,opacity=.5,black!50!white, fill=black!40!white] (G2) -- (orig) -- (Cb) -- cycle;
	\draw[thick,opacity=.5,black!50!white, fill=black!40!white] (4.5,0,.75) -- (0,0,.75) -- (0,2.25,.75) -- cycle;
	\draw[thick,opacity=.5,black!50!white, fill=black!40!white] (0,4.5,.75) -- (0,0,.75) -- (2.25,0,.75) -- cycle;
	\draw[thick,opacity=.5,black!50!white, fill=black!40!white] (4.5,0,0) -- (4.5,0,.75) -- (0,0,.75) -- (orig)--cycle;
  \draw[thick,opacity=.5,black!50!white, fill=black!40!white] (0,4.5,0) -- (0,4.5,.75) -- (0,0,.75) -- (orig)--cycle;
	%%draw the IMAGE of the cube
	
	\draw[thick,red!50!white, fill=red!30!white] (4.5,1.5,0) -- (1.5,1.5,0) -- (1.5,4.5,0) -- (4.5,4.5,0) -- cycle;
	%draw the above of the cube
	\draw[cube] (1.5,1.5,.75) -- (1.5,4.5,.75) -- (4.5,4.5,.75) -- (4.5,1.5,.75) -- cycle;
	
	%%draw the back-right of the cube
	\draw[cube] (1.5,1.5,0) -- (4.5,1.5,0) -- (4.5,1.5,.75) -- (1.5,1.5,.75) -- cycle;

	%%draw the back-left of the cube
	\draw[cube] (1.5,1.5,0) -- (1.5,1.5,.75) -- (1.5,4.5,.75) -- (1.5,4.5,0) -- cycle;
%
  %\draw[cube]  (1.5,1.5,1.5) --(1.5,1.5,3);
	%\draw[green!90!black, thick, dashed]  (0,1.5,3) --(1.5,1.5,3);
  %\draw[green!90!black, thick, dashed]  (1.5,0,3) --(1.5,1.5,3);
	%
   
	%
	 %%draw the IMAGE of the Traingle
	%\draw[thick,opacity=.5,blue!50!white, fill=blue!40!white] (2.25,0,0) -- (.75,2.25,0) -- (3,3,0) -- cycle;
	%
	 %%draw the the Traingle
	%\draw[thick,opacity=.5,blue!50!white, fill=blue!80!white] (2.25,0,2.25) -- (.75,2.25,2.25) -- (3,3,1.5) -- cycle;
	%
	%
	%\draw[blue!50!white, thick, dotted]  (2.25,0,0) --(2.25,0,2.25);
  %\draw[blue!50!white, thick, dotted]  (.75,2.25,0) --(.75,2.25,2.25);
	%\draw[blue!50!white, thick, dotted]  (3,3,0) --(3,3,1.5);

%draw a grid in the x-y plane
	\foreach \x in {0,.75,...,4.5}
		\foreach \y in {0,.75,...,4.5}
		{
			\draw[grid] (\x,0) -- (\x,4.5);
			\draw[grid] (0,\y) -- (4.5,\y);
		}			

	%draw the axes
	\draw[axis] (0,0,0) -- (4.5,0,0) node[left]{$x$};
	\draw[axis] (0,0,0) -- (0,4.5,0) node[anchor=west]{$y$};
	\draw[axis] (0,0,0) -- (0,0,2) node[anchor=west]{$z$};

%\draw [fill=blue!50!white,blue!50!white](2.25,0,2.25) circle (2pt) node [above,blue]{$\boldsymbol{p}$};
%\draw [fill=blue!50!white,blue!50!white](.75,2.25,2.25) circle (2pt) node [above,blue]{$\boldsymbol{r}$};
%\draw [fill=blue!50!white,blue!50!white](3,3,1.5) circle (2pt) node [right,blue]{$\boldsymbol{q}$};

\end{tikzpicture}}
    
		\end{center}
		\caption{The target~$T$ is the green rectangular region,
		the grey region is~$F$, and the pink region is $\pi(T)$.
		}\label{fig:three-Dim-Targ-Second}
\end{figure}

Write $\pi(T)$ for the  projections of $T$  in the $xy$-plane, 
see Figure~\ref{fig:three-Dim-Targ-Second}. 

The following two propositions are syntactically identical to
Proposition~\ref{prop:case-split} and
Proposition~\ref{prop-3dim-fixed-a-point}, although now referring to a
different form of the target set $T$.  While the proof of
Proposition~\ref{prop:case-split} carries over verbatim to the new
setting of Proposition~\ref{prop:case-split-2}, we need to slightly modify the proof of 
Proposition~\ref{prop-3dim-fixed-a-point} in order to prove Proposition~\ref{prop-3dim-fixed-a-point-2}.
\begin{proposition}
Let $S\subseteq\RP^3$ be a 3-simplex.  Then
$T\cap S$ is nonempty if and only if at least one of the following holds:
\begin{enumerate}
\item Some vertex of $S$ lies in $T$.
\item Some bounding edge of $S$ intersects either the face of $T$
  supported by the plane $x=a$ or the face of $T$ supported by the
  plane $y=b$.
\item The bounding edge of $T$ supported by the line $x=a\cap y=b$ intersects $S$.
\end{enumerate}
\label{prop:case-split-2}
\end{proposition}

%Write
%\[F=\{(x,y,z) \in \RP^3 \mid x+ay\leq a(b+1),z\leq c\} \cup 
%\{(x,y,z) \in \RP^3 \mid y+bx\leq b(a+1),z\leq c\}.
%\]

The following Proposition refers to the set $F$ as
defined in (\ref{eq:def-set-F}).
\begin{proposition}\label{prop-3dim-fixed-a-point-2}
  Let $S\subseteq\RP^3$ be a 3-simplex such that $S\cap T$ is
  non-empty, but no bounding edge of $S$ meets $T$.  Then some
  vertex of $S$ lies in $F$.
\end{proposition}
\begin{proof}
  Under the assumptions of this proposition, Items 1 and 2 of
  Proposition~\ref{prop:case-split-2} do not hold. Hence the bounding
  edge of $T$ that is supported by the line segment $x=a \cap y=b$
  meets $S$ at some point not on a bounding edge of $S$.  In
  particular, considering the projection in the $xy$-plane, we have
  that the point $(a,b)$ lies in the interior of $\pi(S)$.

  Now consider the plane in $\RP^3$ affinely spanned by $S$.  Write the equation
  of this plane in the form $z=f(x,y)$ for some affine function $f$.
  From the assumption that no bounding edge of $S$ meets $T$, we
  deduce that $(a,b)$ is the only vertex of the convex set
  $\pi(S) \cap \pi(T)$ at which $f$ is bounded above by $c$.
  It follows that $f$ has positive derivative in the direction of the
  positive $x$-axis and positive $y$-axis.  Hence $f$ is bounded above
  by $c$ on the entire region~$R:=\{(x,y) \in \RP^2:x\leq a,y\leq b\}$.

  Now since $(a,b)$ lies in the interior of $\pi(S)$, there is a
  bounding edge $\boldsymbol{x}\boldsymbol{y}$ of $S$ such that~$\pi(\boldsymbol{x})\pi(\boldsymbol{y})$ 
  meets the region $R$.  By
  Proposition~\ref{prop:decomposition-of-F},
  $\pi(\boldsymbol{x})\pi(\boldsymbol{y})$ either has some endpoint in~$R$
 (say $\pi(\boldsymbol{x})$) or has both endpoints in $\pi(F)$.
  Since $f$ is bounded above by $c$ on $R$, in the first case we have
  that $x_3 \leq c$ and hence $\boldsymbol{x}\in F$.  In the second
  case we have that either~$x_3\leq c$ or~$y_3 \leq c$ and hence
  either $\boldsymbol{x}\in F$ or $\boldsymbol{y}\in F$.
  \end{proof}

  We write separate formulas
  $\varphi_T^{(1)},\varphi_T^{(2)},\varphi_T^{(3)}$, respectively
  expressing the three necessary and sufficient conditions for
  $T\cap S$ to be nonempty as identified in
  Proposition~\ref{prop:case-split-2}.  These are formulas of
  arithmetic whose free variables denote the coordinates of the three
  vertices of $S$.  The definitions of the formulas $\varphi_T^{(1)}$
  and $\varphi_T^{(3)}$ are almost identical to those of the
  corresponding formulas in Section~\ref{sec:three-cost-variables}. The only difference is that for
  $\varphi_T^{(3)}$ we ask to express that the point~$(a,b,c)$ lies
  above the plane affinely spanned by $\boldsymbol{p}$,
  $\boldsymbol{q}$, and $\boldsymbol{r}$ (rather than below the
  plane, as in Section~\ref{sec:three-cost-variables}).

  There are more substantial differences in the definition of the
  formula $\varphi_T^{(2)}$.  Recall that this formula expresses that
  some bounding edge of $S$ meets a face of $T$.  As in
  Section~\ref{sec:three-cost-variables}, it is straightforward to obtain
  $\varphi_T^{(2)}$ given a formula $\psi$ expressing that an
  arbitrary line segment ${\boldsymbol{x}\boldsymbol{y}}$ in $\RP^3$
  meets a given fixed face of $T$.  We outline such a formula
  below. For concreteness we consider the face of $T$ supported by the
  plane $x=a$, which maps under $\pi$ to the line segment~$L$ given by
  $x=a \cap y \geq b$ (see Figure~\ref{fig:linemeetface3d}).  Formula $\psi$ has six free variables,
  respectively denoting the coordinates of~$\boldsymbol{x}$ and
  $\boldsymbol{y}$.

\begin{figure}[H]
 \begin{center}
   \scalebox{.75}{\begin{tikzpicture}
	[grid/.style={very thin,gray},
		axis/.style={->, thick},
		cube/.style={opacity=.5, thick,green,fill=green}]
	
	\coordinate (C) at (1.5,1.5); 
	\coordinate (Ca) at (2.25,0);
	\coordinate (Cb) at (0,2.25);
  \coordinate (orig) at (0,0);
	\coordinate (xax) at (4.5,0);
	\coordinate (yax) at (0,4.5);
	\coordinate (G1) at (intersection of C--Ca and orig--yax);
  \coordinate (G2) at (intersection of C--Cb and orig--xax);
	\coordinate (c) at (1.5,1.5); 
  \coordinate (x) at (.85,3.65); 
	\coordinate (y) at (2.5,1.75); 
	\coordinate (L) at (1.5,5.25);
	
	 %draw the IMAGE of F
	\draw[thick,opacity=.5,black!50!white, fill=black!40!white] (G1) -- (orig) -- (Ca) -- cycle;
	\draw[thick,opacity=.5,black!50!white, fill=black!40!white] (G2) -- (orig) -- (Cb) -- cycle;

%\draw[thick,red!50!white, fill=red!30!white] (0,0) -- (1.5,0) -- (1.5,1.5) -- (0,1.5) -- cycle;

\draw[very thin,color=gray,step=.75cm] (0,0) grid (5.25,5.25);
\draw[->,ultra thick] (-.1,0)--(5.25,0) node[right]{$x$};
\draw[->,ultra thick] (0,-.1)--(0,5.25) node[above]{$y$};
	
\draw [fill=red,red](c) circle (2pt) node [above,red]{$~~\boldsymbol{c}$};
\draw [fill=black,black](x) circle (2pt) node [above,black]{$\pi(\boldsymbol{x})$};
\draw [fill=black,black](y) circle (2pt) node [right,black]{$\pi(\boldsymbol{y})$};
\draw (x) -- (y);
\draw[red,thick] (c) -- (L) node [right,red]{$L$};;

\end{tikzpicture}}
		\end{center}
\caption{To express that $\pi(\boldsymbol{x})\pi(\boldsymbol{y})$ meets line segment $L$.
  The grey region is $\pi(F)$.}
\label{fig:linemeetface3d}
\end{figure}

Formula $\psi$ is a conjunction of two parts.  The first part
expresses that ${\pi(\boldsymbol{x})\pi(\boldsymbol{y})}$ meets~$L$.
The key is to express this requirement via a formula of linear
arithmetic.  For each fixed value of $\pi(\boldsymbol{x}) \in F$ we
can write a linear constraint expressing that
${\pi(\boldsymbol{x})\pi(\boldsymbol{y})}$ meets $L$, and likewise for
each fixed value of $\pi(\boldsymbol{y}) \in F$.  Thus we may assume
that both $\pi(\boldsymbol{x})$ and $\pi(\boldsymbol{y})$ lie in the
complement of $\pi(F)$.  But then
${\pi(\boldsymbol{x})\pi(\boldsymbol{y})}$ meets $L$ just in case
$\pi(\boldsymbol{x})$ and $\pi(\boldsymbol{y})$ lie on opposite sides
of the line $x=a$, which is also a linear constraint.

Suppose now that ${\pi(\boldsymbol{x})\pi(\boldsymbol{y})}$
meets $L$, say at a point
$\pi(\boldsymbol{z})$ where $\boldsymbol{z}$ lies on line 
segment~${\boldsymbol{x}\boldsymbol{y}}$.  The second part of $\psi$
expresses that $\boldsymbol{z}$ lies below the plane $z=c$.  Such a
formula is a disjunction of atoms, each with a
single quadratic term, whose satisfiability is known to be decidable from Theorem~\ref{thm:segal}.

\section{Reachability for Two Observers} \label{subsec:two-cost-variables}
In this section we consider MPTA with two observers  and
reachability of sets of valuations $T\subseteq\RP^{\mathcal{Y}}$
described by arbitrary conjunctions of constraints of the form
$\gamma(y)\sim c$ for $y\in\mathcal{Y}$, ${\sim}\in\{\leq,\geq\}$, and
$c\in\mathbb{Z}$.
  Since the set of valuations in $\RP^{\mathcal{Y}}$
dominating a given valuation can be written in the above form, this
reachability problem subsumes the Pareto Domination Problem.  In
contrast to the situation with three observers, in the case
at hand we will be able to translate the reachability problem into
satisfiability in linear arithmetic.

\subsection{Bounded Cost Objective}

We show how to construct a quantifier-free formula
$\varphi_{\mathrm{Obj}}$ of linear arithmetic that is satisfiable if and only if the
bounded rectangular cost objective can be achieved.

	
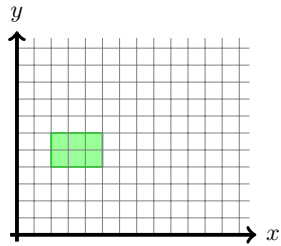
\begin{figure}[H]
 \begin{center}
   \scalebox{.9}{\begin{tikzpicture}[every node/.style={black,above right}]

	\coordinate (t1) at (.5,1); 
	\coordinate (t2) at (1.25,1.5); 
	\draw[thick,green,fill=green!40!white] (t1) rectangle (t2);
	\mygrid{0}{0}{3.5}{3}
\end{tikzpicture}}
    \label{fig:bounded-Cost-Positive-targ}
		\end{center}
\end{figure}

Recall that for a MPTA featuring two non-negative cost variables, a
configuration of the simplex automaton $\mathcal{S}(A)$ determines a
triangle in the plane whose vertices are non-negative integers.  We
denote the vertices $\boldsymbol{p}$, $\boldsymbol{q}$, and $\boldsymbol{r}$.

Draw a line with slope 45 degrees, intersecting the two positive
coordinate axes and passing through the top right corner~$\boldsymbol{x}$ of the
target rectangle $T$.  This line divides the upper right quadrant of
the plane into two regions---a bounded region below the line (shaded
blue) and an unbounded region above the line (shaded grey).  Clearly
the number of vertices of~$\triangle \boldsymbol{p}\boldsymbol{q}\boldsymbol{r}$ that lie in the blue region
is either one, two, or three.  Since the blue region contains finitely
many integer points, the case in which $\triangle \boldsymbol{p}\boldsymbol{q}\boldsymbol{r}$ lies completely
in the blue region is trivial.  The two remaining cases are as
follows:
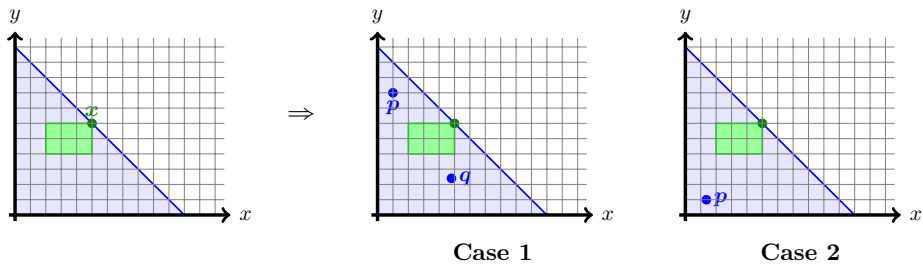
\begin{figure}[H]
 \begin{center}
   \scalebox{.9}{\begin{tikzpicture}[dot/.style={circle,inner sep=1pt,fill,label={#1},name=#1},
  extended line/.style={shorten >=-#1,shorten <=-#1},
  extended line/.default=1cm]

\node [draw=none] (q) at (0,0) {~\scalebox{.9}{\begin{tikzpicture}

	\coordinate (t1) at (.5,1); 
	\coordinate (t2) at (1.25,1.5); 
	\coordinate (dubt2) at (-.25,3);
	\coordinate (orig) at (0,0);
	\coordinate (X0) at (0,4);
	\coordinate (Y0) at (4,0);
	\coordinate (G1) at (intersection of t2--dubt2 and orig--X0);
  \coordinate (G2) at (intersection of t2--dubt2 and orig--Y0);

	\fill[blue!10!white]	(G1) -- (G2) -- (orig)--	(G1);
	\draw[thick,blue] (G1) -- (G2);
	\draw[thick,green,fill=green!40!white] (t1) rectangle (t2);
  \draw [fill=green!50!black,green!50!black](1.25,1.5) circle (2pt) node [above,green!50!black]{$\boldsymbol{x}$};
  \mygrid{0}{0}{3.5}{3};
	
\end{tikzpicture}}};
\node [draw=none] (q) at (2.5,0) {{\Large $\Rightarrow$}};
\node [draw=none,label=below:{{\bf Case 1}}] (q) at (5.3,0) {~\scalebox{.9}{\begin{tikzpicture}[every node/.style={black,above right}]

	\coordinate (t1) at (.5,1); 
	\coordinate (t2) at (1.25,1.5);
	\coordinate (dubt2) at (-.25,3);
	\coordinate (orig) at (0,0);
	\coordinate (X0) at (0,4);
	\coordinate (Y0) at (4,0);
	\coordinate (G1) at (intersection of t2--dubt2 and orig--X0);
  \coordinate (G2) at (intersection of t2--dubt2 and orig--Y0);

\fill[blue!10!white]	(G1) -- (G2) -- (orig)--	(G1);
	\draw[thick,blue] (G1) -- (G2);
	\draw[thick,green,fill=green!40!white] (t1) rectangle (t2);
  \draw [fill=green!50!black,green!50!black](1.25,1.5) circle (2pt) node [above,green!50!black]{};
\draw [fill=blue,blue](1.2,.6) circle (2pt) node [right,blue]{$\boldsymbol{q}$};
 \draw [fill=blue,blue](.25,2) circle (2pt) node [below,blue]{$\boldsymbol{p}$};

		\mygrid{0}{0}{3.5}{3}

\end{tikzpicture}}};
\node [draw=none,label=below:{{\bf Case 2}}] (q) at (9.8,0) {~\scalebox{.9}{\begin{tikzpicture}

 \coordinate (t1) at (.5,1); 
	\coordinate (t2) at (1.25,1.5); 
	\coordinate (dubt2) at (-.25,3);
	\coordinate (orig) at (0,0);
	\coordinate (X0) at (0,4);
	\coordinate (Y0) at (4,0);
	\coordinate (G1) at (intersection of t2--dubt2 and orig--X0);
  \coordinate (G2) at (intersection of t2--dubt2 and orig--Y0);

\fill[blue!10!white]	(G1) -- (G2) -- (orig)--	(G1);
	\draw[thick,blue] (G1) -- (G2);
	\draw[thick,green,fill=green!40!white] (t1) rectangle (t2);
  \draw [fill=green!50!black,green!50!black](1.25,1.5) circle (2pt) node [above,green!50!black]{};

\draw [fill=blue,blue](.342,.25) circle (2pt) node [right,blue]{$\boldsymbol{p}$};

\mygrid{0}{0}{3.5}{3}
\end{tikzpicture}}};

\end{tikzpicture}}
    \label{fig:bounded-Cost-Positive}
		\end{center}
\end{figure}

\begin{itemize} 
\item {\bf Case 1:} the blue region contains two vertices of
  $\triangle \boldsymbol{p}\boldsymbol{q}\boldsymbol{r}$---say $\boldsymbol{p}$ and $\boldsymbol{q}$.  We proceed by a case analysis on
  the coordinates of $\boldsymbol{p}$ and $\boldsymbol{q}$ (for which there are finitely many
  possibilities).  Fix values for $\boldsymbol{p}$ and $\boldsymbol{q}$ in the blue region.
  Then the condition that $\triangle \boldsymbol{p}\boldsymbol{q}\boldsymbol{r}$ intersects the target can be
  written as a linear constraint on the coordinates of the remaining vertex~$\boldsymbol{r}$---specifically 
that one of the vertices of $\triangle \boldsymbol{p}\boldsymbol{q}\boldsymbol{r}$ lies
  in the target $T$ or that one of the bounding line segments of
  $\triangle \boldsymbol{p}\boldsymbol{q}\boldsymbol{r}$ intersects one of the bounding line segments of the
  target $T$.

\item {\bf Case 2:} the blue region contains a single vertex of
  $\triangle \boldsymbol{p}\boldsymbol{q}\boldsymbol{r}$---say $\boldsymbol{p}$.  Fix a value of $\boldsymbol{p}$ and assume that $\boldsymbol{p}$
  is not in the target $T$.  Now consider the ``shadow'' of the target
  rectangle~$T$ created by a light source at point~$\boldsymbol{p}$ (the pink
  region in the diagram).  This shadow is is a region in the plane
  that is bounded by two lines that respectively pass through~$\boldsymbol{p}$ and
  vertices of the target $T$ (shown as pink dashed lines in the
  diagram).  Then in case vertices $\boldsymbol{q}$ and $\boldsymbol{r}$ lie in the grey region,
  $\triangle \boldsymbol{p}\boldsymbol{q}\boldsymbol{r}$ fails to meet the target rectangle if and only~$\boldsymbol{q}$
  and~$\boldsymbol{r}$ both lie on the same side of both of the pink dashed lines.
  Again this condition can be expressed as a Boolean combination of
  linear constraints on~$\boldsymbol{q}$ and~$\boldsymbol{r}$ since the pink dashed lines are
  fixed.

	
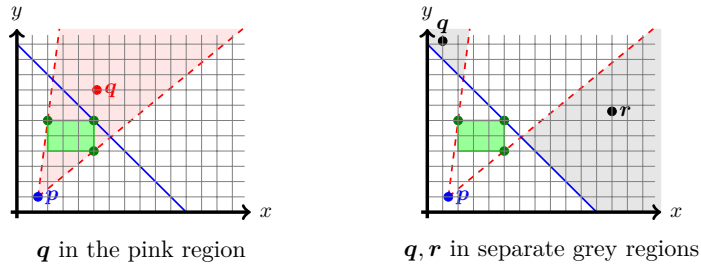
\begin{figure}[H]
 \begin{center}
   \scalebox{.9}{\begin{tikzpicture}[dot/.style={circle,inner sep=1pt,fill,label={#1},name=#1},
  extended line/.style={shorten >=-#1,shorten <=-#1},
  extended line/.default=1cm]

\node [draw=none,label=below:{$\boldsymbol{q}$ in the pink region}] (q) at (0,0) {~\scalebox{.9}{\begin{tikzpicture}[every node/.style={black,above right}]

\coordinate (t1) at (.5,1); 
\coordinate (t2) at (1.25,1.5); 
\coordinate (A) at	(.342,.25);
\coordinate (P1) at (.5,1.5); 
\coordinate (P2) at	(1.25,1);
\coordinate (P3) at (0,3);
\coordinate (P4) at (4.5,3);
\coordinate (C1) at (intersection of A--P1 and P3--P4);
\coordinate (C2) at (intersection of A--P2 and P3--P4);

 \coordinate (dubt2) at (-.25,3);
	\coordinate (orig) at (0,0);
	\coordinate (X0) at (0,4);
	\coordinate (Y0) at (4,0);
	\coordinate (G1) at (intersection of t2--dubt2 and orig--X0);
  \coordinate (G2) at (intersection of t2--dubt2 and orig--Y0);
%\coordinate (D1) at (intersection of G1--G2 and A--P1);
%\coordinate (D2) at (intersection of G1--G2 and A--P2);

\fill[red!10!white]	(C1)--(A)--(C2)--(C1);
	
	\draw[thick,red, dashed,name path=line 1] (A) -- (P1)--(C1);
	\draw[thick,red, dashed,name path=line 1] (A) -- (P2)--(C2);
	
\draw[thick,blue] (G1) -- (G2);	
	
\draw[thick,green,fill=green!40!white] (t1) rectangle (t2);
\draw [fill=blue,blue](A) circle (2pt) node [right,blue]{$\boldsymbol{p}$};
\draw [fill=red,red](1.3,2) circle (2pt) node [right,red]{$\boldsymbol{q}$};

\draw [fill=green!50!black,green!50!black](P1) circle (2pt) node [left,green!50!black]{};
\draw [fill=green!50!black,green!50!black](P2) circle (2pt) node [below,green!50!black]{};
\draw [fill=green!50!black,green!50!black](t2) circle (2pt) node [above,green!50!black]{};

\mygrid{0}{0}{3.8}{3}

\end{tikzpicture}}};
\node [draw=none,label=below:{$\boldsymbol{q},\boldsymbol{r}$ in separate grey regions}] (q) at (6,0) {~\scalebox{.9}{\begin{tikzpicture}[every node/.style={black,above right}]
\coordinate (t1) at (.5,1); 
\coordinate (t2) at (1.25,1.5); 
\coordinate (A) at	(.342,.25);
\coordinate (P1) at (.5,1.5); 
\coordinate (P2) at	(1.25,1);
\coordinate (P3) at (0,3);
\coordinate (P4) at (3.8,3);
\coordinate (C1) at (intersection of A--P1 and P3--P4);
\coordinate (C2) at (intersection of A--P2 and P3--P4);

\coordinate (dubt2) at (-.25,3);
	\coordinate (orig) at (0,0);
	\coordinate (X0) at (0,4);
	\coordinate (Y0) at (4,0);
	\coordinate (G1) at (intersection of t2--dubt2 and orig--X0);
  \coordinate (G2) at (intersection of t2--dubt2 and orig--Y0);

\coordinate (D1) at (intersection of G1--G2 and A--P1);
\coordinate (D2) at (intersection of G1--G2 and A--P2);

	\fill[black!10!white] (P3)--(G1)--(D1)--(C1)--(P3);
	\fill[black!10!white] (G2)--(D2)--(C2)--(P4)--(3.8,0)--(G2);

	\draw[thick,red, dashed] (A) -- (P1)--(C1);
	\draw[thick,red, dashed] (A) -- (P2)--(C2);
	\draw[thick,blue] (G1)--(G2);
\draw[thick,green,fill=green!40!white] (t1) rectangle (t2);
\draw [fill=blue,blue](A) circle (2pt) node [right,blue]{$\boldsymbol{p}$};

\draw [fill=black,black](.25,2.8) circle (2pt) node [above,black]{$\boldsymbol{q}$};
\draw [fill=black,black](3,1.65) circle (2pt) node [right,black]{$\boldsymbol{r}$};

\draw [fill=green!50!black,green!50!black](P1) circle (2pt) node [left,green!50!black]{};
\draw [fill=green!50!black,green!50!black](P2) circle (2pt) node [right,green!50!black]{};
\draw [fill=green!50!black,green!50!black](t2) circle (2pt) node [above,green!50!black]{};

\mygrid{0}{0}{3.8}{3}

\end{tikzpicture}}};

\end{tikzpicture}}
    \label{fig:bounded-Cost-Positive-case2}
		\end{center}
\end{figure}

\end{itemize}

\subsection{Unbounded Cost Objective}

We show how to construct a quantifier-free formula
$\varphi_{\mathrm{Obj}}$ of linear arithmetic that is satisfiable if and only if the
unbounded rectangular cost objective, as shown in the diagram below, can be achieved.
We consider an objective where the observer~$x$ is unbounded  above while~$y$ is bounded.
The  case when $x$ is bounded with $y$ unbounded   above is symmetric.
The last case for an unbounded cost objective is when both observers~$x,y$ are unbounded   above.
The following argument can be used in this last case with a slight modification.

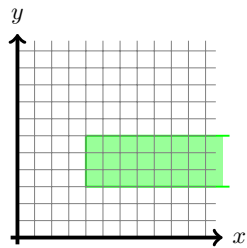
\begin{figure}[H]
 \begin{center}
   \scalebox{.9}{\begin{tikzpicture}

	\coordinate (t1) at (1,.75); 
	\coordinate (t2) at (1,1.5); 
	\coordinate (t3) at (3,.75);
	\fill[green!40!white] (t2) rectangle (t3);
\draw[thick,green] (3.1,.75)--(t1)--(t2)--(3.1,1.5);
  \mygrid{0}{0}{3}{3};
	
\end{tikzpicture}}
    \label{fig:unbounded-Cost-Positive-targ}
		\end{center}
\end{figure}

Draw a line with slope 45 degrees, intersecting the two positive
coordinate axes and passing through the top left corner~$P$
of the target rectangle $T$.
 This line
divides the upper right quadrant of the plane into two regions---a
bounded region below the line (shaded blue) and an unbounded region
above the line.  We further divide the region above the line into
three horizontal bands with boundaries given by the horizontal sides
of the target (the upper bound is shaded pink and lower band is shaded
grey in the diagram).

  We now consider two cases according to whether $\triangle \boldsymbol{p}\boldsymbol{q}\boldsymbol{r}$ has a
  vertex in the blue region.

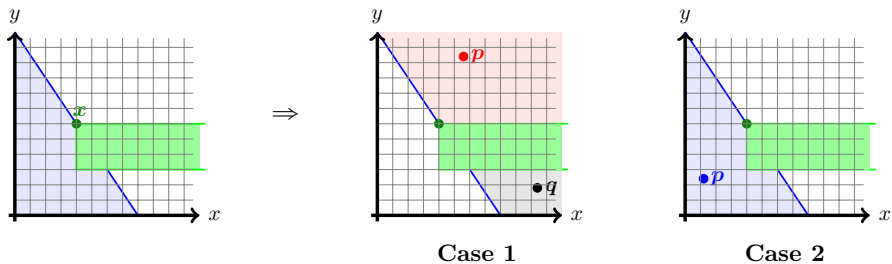
\begin{figure}[H]
 \begin{center}
   \scalebox{.9}{\begin{tikzpicture}[dot/.style={circle,inner sep=1pt,fill,label={#1},name=#1},
  extended line/.style={shorten >=-#1,shorten <=-#1},
  extended line/.default=1cm]

\node [draw=none] (q) at (0,0) {~\scalebox{.9}{\begin{tikzpicture}

	\coordinate (t1) at (1,.75); 
	\coordinate (t2) at (1,1.5); 
	\coordinate (t3) at (3,.75);
	\coordinate (dubt2) at (0,3);
	\coordinate (orig) at (0,0);
	\coordinate (X0) at (0,2);
	\coordinate (Y0) at (3,0);
	\coordinate (G1) at (intersection of t2--dubt2 and orig--X0);
  \coordinate (G2) at (intersection of t2--dubt2 and orig--Y0);

  \fill[blue!10!white]	(G1) -- (G2) -- (orig)--	(G1);
	\draw[thick,blue] (G1) -- (G2);
	\fill[green!40!white] (t2) rectangle (t3);
\draw[thick,green] (3.1,.75)--(t1)--(t2)--(3.1,1.5);
  \draw [fill=green!50!black,green!50!black](t2) circle (2pt) node [above,green!50!black]{$~\boldsymbol{x}$};

  \mygrid{0}{0}{3}{3};
\end{tikzpicture}}};
\node [draw=none] (q) at (2.5,0) {{\Large $\Rightarrow$}};
\node [draw=none,label=below:{{\bf Case 1}}] (q) at (5.3,0) {~\scalebox{.9}{\begin{tikzpicture}
 		\coordinate (t1) at (1,.75); 
	\coordinate (t2) at (1,1.5); 
	\coordinate (t3) at (3,.75);
	\coordinate (dubt2) at (0,3);
	\coordinate (orig) at (0,0);
	\coordinate (X0) at (0,2);
	\coordinate (Y0) at (3,0);
	\coordinate (G1) at (intersection of t2--dubt2 and orig--X0);
  \coordinate (G2) at (intersection of t2--dubt2 and orig--Y0);
  
  \fill[red!10!white]	(G1)--(t2)--(t3)--(3,3)--(G1);
  \fill[black!10!white]	(G2)--(1.5,.75)--(t3)--(3,0)--(G2);
	\draw[thick,blue] (G1) -- (G2);
	\draw[thick,green] (3.1,.75)--(t1)--(t2)--(3.1,1.5);
	\fill[green!40!white] (t2) rectangle (t3);

	\draw [fill=red,red](1.4,2.6) circle (2pt) node [right,red]{$\boldsymbol{p}$};
	\draw [fill=black,black](2.6,.45) circle (2pt) node [right,black]{$\boldsymbol{q}$};
	 \draw [fill=green!50!black,green!50!black](t2) circle (2pt) node [above,green!50!black]{};
\mygrid{0}{0}{3}{3}
\end{tikzpicture}}};
\node [draw=none,label=below:{{\bf Case 2}}] (q) at (9.8,0) {~\scalebox{.9}{\begin{tikzpicture}[every node/.style={black,above right}]

  	\coordinate (t1) at (1,.75); 
	\coordinate (t2) at (1,1.5); 
	\coordinate (t3) at (3,.75);
	\coordinate (dubt2) at (0,3);
	\coordinate (orig) at (0,0);
	\coordinate (X0) at (0,2);
	\coordinate (Y0) at (3,0);
	\coordinate (G1) at (intersection of t2--dubt2 and orig--X0);
  \coordinate (G2) at (intersection of t2--dubt2 and orig--Y0);
  \coordinate (A) at (1,.75); 

  \fill[blue!10!white]	(G1) -- (G2) -- (orig)--	(G1);
	\draw[thick,blue] (G1) -- (G2);
	\draw[thick,green] (3.1,.75)--(t1)--(t2)--(3.1,1.5);
	\fill[green!40!white] (t2) rectangle (t3);
 \draw [fill=green!50!black,green!50!black](t2) circle (2pt) node [above,green!50!black]{};
	\draw [fill=blue,blue](.3,.6) circle (2pt) node [right,blue]{$\boldsymbol{p}$};
  \mygrid{0}{0}{3}{3};

\end{tikzpicture}}};

\end{tikzpicture}}
    \label{fig:unbounded-Cost-Positive}
		\end{center}
\end{figure}
\begin{itemize}
\item {\bf Case 1.} No vertex of $\triangle \boldsymbol{p}\boldsymbol{q}\boldsymbol{r}$ lies in the blue
  region.  Then $\triangle \boldsymbol{p}\boldsymbol{q}\boldsymbol{r}$ meets the target iff it is not the
  case that all vertices lie in the grey region or all vertices lie in
  the pink region.
\item {\bf Case 2.} Some vertex of $\triangle \boldsymbol{p}\boldsymbol{q}\boldsymbol{r}$ lies in the blue
  region---say $\boldsymbol{p}$.  Fix $\boldsymbol{p}$.  Then $\triangle \boldsymbol{p}\boldsymbol{q}\boldsymbol{r}$ meets $T$ if one
  of the line segments $\overline{\boldsymbol{p}\boldsymbol{q}}$ or $\overline{\boldsymbol{p}\boldsymbol{r}}$ intersects
  the boundary of the target~$T$.  Given that~$\boldsymbol{p}$ is fixed this
  condition can be expressed as a Boolean combination of linear
  constraints on~$\boldsymbol{q}$ and $\boldsymbol{r}$.
\end{itemize}

%\subsection{Unbounded cost objective}
%
%
%\begin{figure}[H]
 %\begin{center}
   %\scalebox{.9}{\input{Figures/unboundedCostTarg.tex}}
    %\label{fig:unbounded-Cost-targ}
		%\end{center}
%\end{figure}
%
%There are cases:
%
%\begin{figure}[H]
 %\begin{center}
   %\scalebox{.9}{\input{Figures/unboundedCost.tex}}
    %\label{fig:unbounded-Cost}
		%\end{center}
%\end{figure}

%%%%%%%%%%%%%%%%%%%%%%%%%%%%%%%%%%%%%%%%%%%%%%%%%%%%%%%%%%%%%%%%%%%%

%The key technical lemma is as follows.
%\begin{lemma}
%Let $L\subseteq \mathbb{Z}^6$ be a linear set and let $C\subseteq
%\mathbb{R}^2$ be an affine cone given as the conjunction of two linear
%inequalities with rational coefficients.  Then it is decidable whether
%there exists a point $(x_1,x_2,y_1,y_2,z_1,z_2)\in L$ such that $C$
%meets the triangle with vertices
%$X:=(x_1,x_2)$, $Y:=(y_1,y_2)$, and $Z:=(z_1,z_2)$.
%%with rational endpoint and slope,
%%whether there exists $(x_1,x_2,y_1,y_2) \in L$ such that $H$ meets the
%%line segment $\overline{XY}$ with endpoints $X=(x_1,x_2)$ and
%%$Y=(y_1,y_2)$.
%\end{lemma}
%\begin{proof}
%Either $X$ lies in the cone or one of the line segments $\overline{XY}$
%or $\overline{XZ}$ intersects one of the extremal rays of $C$.
%\end{proof}

\end{document}